\documentclass[10pt]{article}
\usepackage{fullpage}
\usepackage{amsmath,amsthm,amsfonts,bm}
\usepackage{setspace}
\usepackage{algorithm,algorithmic}
\usepackage[obeyFinal, textsize=tiny]{todonotes} 

\usepackage[colorlinks,urlcolor=blue,citecolor=blue,linkcolor=blue]{hyperref}

\newcommand{\keti}[1]{\left|{\psi_{#1}}\right\rangle}

\newcommand{\lloyd}{Lloyd et al.~\cite{lloyd2016quantum}~}
\usepackage{color}
\usepackage{amsmath}
\usepackage{amsfonts}
\usepackage{amssymb}
\usepackage{amsmath}
\usepackage[braket, qm]{qcircuit}
\newtheorem{theorem}{Theorem}
\newtheorem{lemma}{Lemma}
\newtheorem{remark}{Remark}

\DeclareFontFamily{U}{wncy}{}
\DeclareFontShape{U}{wncy}{m}{n}{<->wncyr10}{}
\DeclareSymbolFont{mcy}{U}{wncy}{m}{n}
\DeclareMathSymbol{\Sh}{\mathord}{mcy}{"58} 
\newcommand{\RR}{\mathbb{R}}
\newcommand{\tH}{\tilde{\mathcal{H}}}
\newcommand{\tpar}{\tilde{\partial}}
\newcommand{\trace}{\mathrm{trace}}
\newcommand{\nv}{\mathrm{n_{v}}}
\newcommand{\nz}{\mathrm{nnz}}
\newcommand{\veps}{\varepsilon}

\newcommand{\eps}{\epsilon}
\newcommand{\PG}{P_{\Gamma}}
\newcommand{\defeq}{\stackrel{\textit{\tiny{def}}}{=}}
\newcommand{\img}{\mathrm{img}}
\newcommand{\rank}{\mathrm{rank}}
\newcommand{\bdiag}{\mathrm{Blockdiag}}
\newcommand{\pr}{\mathbb{P}}
\parskip = \baselineskip

\title{
Quantum Topological Data Analysis \\ with Linear Depth and Exponential
Speedup
}

 \author{Shashanka Ubaru\thanks{Equal Contribution.} \thanks{IBM Research,
 USA.} \and Ismail Yunus Akhalwaya\footnotemark[1] \thanks{IBM Research, South Africa. School of Computer Science and Applied Mathematics, University of the Witwatersrand, Johannesburg, South Africa.}  \and Mark S.\ Squillante\footnotemark[2] \and Kenneth L.\ Clarkson\footnotemark[2]
 \and Lior Horesh\footnotemark[2]}
 \date{}
\begin{document}

\maketitle

\begin{abstract}
Quantum computing offers the potential of exponential speedups for certain classical computations. Over the last decade, many quantum machine learning (QML) algorithms have been proposed as candidates for such exponential improvements. However, two issues unravel the hope of exponential speedup for some of these QML algorithms: the data-loading problem and, more recently, the stunning ``dequantization'' results of
Tang et~al. A third issue, namely the fault-tolerance requirements of most QML algorithms, has further hindered their practical realization.
The quantum topological data analysis (QTDA) algorithm of Lloyd, Garnerone and Zanardi was one of the first QML algorithms that convincingly offered an expected exponential speedup. From the outset, it did \emph{not} suffer from the data-loading problem.
A recent result has also shown that the generalized problem solved by this algorithm is likely classically intractable, and would therefore be immune to any dequantization efforts.
However, the QTDA algorithm of Lloyd et~al.\ has a time complexity of  $O(n^4/(\eps^2 \delta\sqrt{\zeta}))$ (where $n$ is the number of data points, $\eps$ is the error tolerance, $\delta$ is the smallest nonzero eigenvalue of the restricted Laplacian, and $\zeta$ is the fraction of all simplices in the complex) and requires fault-tolerant quantum computing, which has not yet been achieved. In this paper, we completely overhaul the QTDA algorithm to achieve an improved exponential speedup and depth complexity of ${O}(n\log(1/(\delta\eps)))$.
The latter depth complexity opens the door for an implementation on near-term quantum hardware, potentially making it the first useful  algorithm to achieve quantum advantage on general classical data.
Our approach includes three key innovations: (a) an efficient realization of the combinatorial Laplacian as a sum of Pauli operators; (b) a quantum rejection sampling and projection approach to restrict the superposition to the simplices of the desired order in the complex (replacing Grover's search of Lloyd et al.); and (c) a stochastic rank estimation method to estimate the Betti numbers (replacing quantum phase estimation of Lloyd et al.). We present a theoretical error analysis for the proposed algorithm, and present the circuit and computational time and depth complexities for Betti number estimation up to
the error tolerance $\eps$.
The techniques presented herein have wider potential applications than QTDA or even rank estimation.
\end{abstract}


\section{Introduction}\label{sec:intro}
With the realization of the computational power of quantum computing in the 1980s, there has been an active search for algorithms that exploit this power to achieve exponential  speedups  over  classical
(i.e., non-quantum)
algorithms on digital computers. Since the power of quantum computing can be characterized as the ability to perform certain linear-algebraic operations in 
exponentially large spaces, several algorithms have been proposed for quantum machine learning (QML) from the perspective of linear algebra~\cite{biamonte2017quantum,schuld2015introduction,schuld2019quantum}. 
Learning models based on parameterized quantum circuits have also been proposed~\cite{havlivcek2019supervised,benedetti2019parameterized,beer2020training,liu2020rigorous}. These QML methods were suggested as potential examples of useful quantum applications that achieve polynomial to exponential speedups over the best-known classical methods for machine learning
.\footnote{Here, by \emph{exponential speedup}, we informally mean that, given classical algorithms taking time $T$ for a problem, there is a quantum algorithm
that takes time $O(\log^cT)$, for a constant $c$, and thus the classical runtime is
essentially
exponential in the quantum runtime. A \emph{super-polynomial speedup} implies that the quantum algorithm takes time smaller than $T^c$ for any fixed $c>0$.}

Despite this great potential, theoretical and practical hurdles remain. 
Two of the most serious concerns are the data-loading\footnote{The \emph{data-loading} problem refers to the case where, although a quantum computer might take time $O(\log(N))$ to solve a particular problem on input of size $N$, the time needed for the
(classical)
data to be set up in
/ ported to
the quantum system might still well be $\Omega(N)$. Sometimes the QRAM proposal may help, but then this triggers three other issues: QRAM may be considered a ``cheat'' since ``exponential'' hardware is still required; the scenarios where QRAM assumptions apply might render the setup prone to dequantization; and, most troubling from a practical perspective, QRAM coupled to real quantum computers are not yet available nor are they expected to be for the near-term future.} problem~\cite{aaronson2015read} and, more dramatically, classical dequantization. Recent developments of ``dequantized'' classical algorithms, which operate under analogous assumptions~\cite{tang2019quantum,chia2020sampling,chepurko2020quantum}, have reduced the potential speedups of many of these linear-algebraic QML  proposals to be at most polynomial.
A third issue, which is a major focus of our attention, is the common requirement of fault-tolerance for QML algorithms. However, fault-tolerance  has not yet been achieved on currently available quantum devices, and we are likely several years away from its realization. Of course, even if fault-tolerance becomes available, it still is advantageous to have runtimes that are as small as possible. 
Moreover, the previous parameterized circuit-based methods~\cite{havlivcek2019supervised,benedetti2019parameterized,beer2020training,liu2020rigorous} are limited to showing quantum speedup only for a specific family of
hand-crafted
data-sets that are likely not practically motivated.

With these considerations, the first big question which arises concerns whether we can develop any QML algorithms that have circuit depth linear in the number of data points $n$
and that ``provably''
achieve the tantalizing promise of asymptotic exponential speedup for general classical data. If we then go further and bring in the challenge of achieving ``quantum advantage'' in the practical non-fault-tolerant regime, the question is no longer
only
about asymptotics but
also about
the realization of a useful algorithm under the constraints of noise and a fixed constant-depth. Of course, ${O}(\log n)$ and maybe ${O}(n)$ run-time algorithms have a fighting chance
in this regard, but in the end it becomes a question of actual circuit depths versus current hardware performance numbers. The term adopted by the research community to describe current hardware performance levels is Noisy Intermediate-Scale Quantum (NISQ) devices~\cite{preskill2018quantum}. Achieving quantum advantage on NISQ devices for a useful algorithm is recognized as a very important goal for near-term commercial viability, enabling a more reliable
path towards full fault-tolerant quantum computing.

\paragraph{Topological Data Analysis (TDA):}
TDA has garnered significant interest in the applied mathematics community~\cite{zomorodian2005computing,ghrist2008barcodes,bubenik2015statistical,wasserman2018topological}. It is one of the few ``big data'' algorithms that can consume massive datasets and reduce them to a handful of global and interpretable signature numbers, laden with predictive and analytical value. Unfortunately, classical
algorithms for TDA have significant computational demands that are currently addressed by sampling and by avoiding high-order Betti numbers. 
Recently, in a
seminal
article, \lloyd proposed a quantum algorithm for TDA (QTDA) that  achieves an  expected  exponential  speedup  under certain conditions. Moreover, the method does not suffer from the data loading problem. Follow-up articles~\cite{gunn2019review,gyurik2020towards} have studied the Lloyd et al.\ algorithm in detail. 
In particular, Gyurik et al.~\cite{gyurik2020towards} showed that the QTDA algorithm solves a problem which is likely classically intractable by establishing that a generalization of the problem is as hard as simulating the one clean qubit model; i.e., it is DQC1-hard~\cite{cade2017quantum}. Classical computers are believed to require super-polynomial time to simulate DQC1, and hence it is argued that the QTDA algorithm enjoys exponential speedup that is likely to be immune to dequantization~\cite{gyurik2020towards}.

The QTDA algorithm due to \lloyd involves two main steps, namely: (i) repeatedly constructing the simplices in the Vietoris-Rips simplicial complex as a mixed state;\footnote{In our understanding, the mixed state is actually unnecessary, as a single simplex can be efficiently drawn classically.} and (ii) projecting this onto the kernel of a simplicial-complex-restricted combinatorial Laplacian operator in order to calculate the dimension of the kernel (see Section~\ref{ssec:tda} for the corresponding definitions). The estimation of the kernel dimension allows the calculation of the topological global invariants known as \emph{Betti numbers}~\cite{friedman1998computing}  (which represent  the  number  of  ``non  filled-in'' structures such as clusters, loops, voids, and so on). The first step is achieved using Grover's search algorithm~\cite{boyer1998tight}, and the second step (kernel dimension estimation) is computed using quantum phase estimation (QPE)~\cite{nielsen2010quantum} of an operator involving the restricted Laplacian.\footnote{Restricting the Laplacian again requires Grover's search.}
Notwithstanding its computational advantages, the QTDA approach of \lloyd strongly requires fault-tolerant quantum computing.
This follows not only from the high-order polynomial complexity, i.e., ${O}(n^4/(\delta\sqrt{\zeta}))$ where $\delta$ denotes the smallest nonzero eigenvalue of the restricted Laplacian and $\zeta$ is the fraction of all simplices in the complex,
resulting in very deep circuits. But also from the fact that Grover's search and QPE, as used in QTDA, are not robust to errors. In particular, both sub-algorithms require precise phase information where any errors would accumulate multiplicatively (leading to catastrophic information loss without the possibility of errors between runs averaging out).

\paragraph{Our Contributions:}
We present here, to the best of our knowledge, the first QML algorithm with
linear-depth complexity
and provable\footnote{Under the generally accepted belief that DQC1 is classically intractable.} exponential speedup on arbitrary input, possibly opening the doors to the first generically useful NISQ algorithm with quantum advantage.
In particular, we present \emph{NISQ-QTDA}, a quantum topological data analysis algorithm that has an improved exponential speedup and
a depth complexity of ${O}(n\log(1/(\delta\eps)))$, where $\eps$ denotes the error tolerance.
This is achieved via three key innovations, namely: (a) an efficient representation of a boundary operator as a sum of Pauli operators; (b) a quantum rejection sampling technique to build the relevant simplicial complex; and (c) a stochastic rank estimation method to estimate the Betti numbers that does not require Quantum Phase Estimation (QPE), i.e., the eigen-decomposition of the combinatorial Laplacian. 
Our approach
facilitates the elimination of
the fault-tolerance requirement based on a couple of key observations.
First, while QPE provides estimates of the actual eigenvalues, we only need the dimension of the kernel, that is, the count of the zero eigenvalues, and thus we are able to replace QPE with an alternative method for this dimension/eigencount estimation that has ${O}(n\log(1/(\delta\eps)))$ depth. 
Second, Grover's search algorithm, which is used to build the simplicial complex and restrict the Laplacian, is quite expensive depth-wise and achieves only a quadratic speedup. We are able to design an alternative scheme based on projections and rejection sampling that
has ${O}(n)$ depth, but that admittedly does require more trials. Excitingly, the technology making this projection approach possible on current quantum computing hardware, namely mid-circuit measurements, has only just recently become available.

\vskip-0.1in
\paragraph{Outline:}
The remainder of the paper is organized as follows.
We first provide in Section~\ref{sec:prelim} some technical preliminaries, including an introduction to the relevant TDA concepts and the QTDA algorithm of Lloyd et al.~\cite{lloyd2016quantum}.
Our proposed \emph{NISQ-QTDA} algorithm and related technical details are presented in Section~\ref{sec:nisq}. Section~\ref{sec:analysis} presents the error and computational complexity analyses of our algorithm, where we also discuss the different settings under which the proposed algorithm has NISQ implementation and achieves improved exponential speedup.
We conclude with some final remarks.


\section{Preliminaries}\label{sec:prelim}
This section provides various technical background and preliminaries. We start by introducing some of the relevant  linear-algebraic concepts of TDA, and then present the quantum TDA algorithm of \lloyd.

\subsection{Topological Data Analysis}\label{ssec:tda}

Given a set of  data-points in some ambient (possibly high-dimensional) space, TDA aims to extract a small set of robust and interpretable features that capture the ``shape'' of the dataset~\cite{zomorodian2005computing,bubenik2015statistical,wasserman2018topological}. \emph{Persistent Homology}~\cite{ghrist2008barcodes} aims to compute these features at different scales of resolution, thus obtaining a set of local and global features to describe the data distribution topologically.
These topological features are not only invariant under rotation and translation, but are also generically robust with respect to the data representation, the data sampling procedure, and noise.
For these reasons, TDA has become a powerful tool used in many data applications~\cite{wasserman2018topological}.
  
The theory and applications of TDA have been extensively studied in the research literature. Here we will focus on the linear-algebraic concepts of TDA that are relevant for the development of our NISQ-QTDA algorithm. We begin with the concept of a simplical complex that is derived from data-points embedded in some ambient space. 
A  $k$-simplex is a collection of $k+1$ vertices forming a simple polytope of dimension~$k$; e.g., $0$-simplices are single points (zero-dimensional), $1$-simplices are line segments (one-dimensional), $2$-simplices are triangles (two-dimensional), and so on.
A simplicial complex is a collection of such simplices (of any order), closed under adding all \emph{lower} simplices, which are simplices obtained by removing one vertex (e.g., if a triangle is in a complex, then all three associated edge simplices are also in the complex and, recursively, so are all three associated points). Homology provides us with a linear-algebraic
approach to extract, from simplicial complexes derived from the data, features that
describe the ``shape'' of the data, such as the number of connected components,
or holes (as in doughnuts), or voids (as in swiss cheese), or
higher-dimensional holes/cavities.

Given a set of~$n$ data-points $\{x_i\}_{i=0}^{n-1}$ in some space together with a distance metric $\mathcal{D}$, a Vietoris-Rips~\cite{ghrist2008barcodes} simplicial complex is constructed by selecting a resolution/grouping scale $\veps$ that defines the ``closeness'' of the points with respect to the distance metric $\mathcal{D}$, and then connecting the points that are a distance of $\veps$ from  each  other (i.e., connecting points $x_i$ and $x_j$ whenever $\mathcal{D}(x_i,x_j)\leq \veps$, forming a so-called 1-skeleton).
A $k$-simplex is then added for every subset of $k+1$ data-points that are pair-wise connnected (i.e., for every $k$-clique, the associated $k$-simplex is added). The resulting simplicial complex is related to the clique-complex from graph theory~\cite{gyurik2020towards}.

Let $S_k$ denote the set of $k$-simplices in the  Vietoris–Rips complex $\Gamma=\{S_k\}_{k=0}^{n-1}$, with $s_k\in S_k$ written as $[j_0,\ldots,j_k]$ where $j_i$ is the $i$th vertex of $s_k$.
Let $\mathcal{H}_k$ denote an ${n \choose k+1}$-dimensional Hilbert space,
with basis vectors corresponding to each of the possible $k$-simplices (all subsets of size $k+1$). Further let $\tH_k$ denote the subspace of $\mathcal{H}_k$ spanned by the basis
vectors corresponding to the simplices in $S_k$,
and let $\ket{s_k}$ denote the basis state corresponding to $s_k\in S_k$.
Then, the $n$-qubit Hilbert space $\mathbb{C}^{2^n}$ is given by
$\mathbb{C}^{2^n}\cong  \bigoplus_{k=0}^n \mathcal{H}_k$.
The boundary map (operator) on $k$-dimensional simplices $ \partial_k:  \mathcal{H}_k \rightarrow {\mathcal{H}}_{k-1}$ is a linear operator
defined by its action on the basis states as follows:
\begin{align}
  \partial_k  \ket{s_k} &= \sum_{l=0}^{k-1} (-1)^l \ket{s_{k-1}(l)} ,
\end{align}
where $\ket{s_{k-1}(l)}$ is the \emph{lower} simplex obtained by leaving out vertex~$l$ (i.e., $s_{k-1}$ has the same vertex set as $s_{k}$ except without $j_l$). Naturally, $s_{k-1}$ is $k-1$-dimensional, one 
dimension
less than $s_{k}$. The factor $(-1)^l$  produces the so-called \emph{oriented}~\cite{ghrist2008barcodes} sum of boundary simplices, which keeps track of neighbouring simplices so that $\partial_{k-1}\partial_{k} \ket{s_k} = 0$, given that the boundary of the boundary is empty.

The boundary map $\tpar_k:  \tH_k \rightarrow \tH_{k-1}$ restricted to a given  Vietoris–Rips complex $\Gamma$  is given by $\tpar_k = \partial_k\tilde{P}_k$, where $\tilde{P}_k$ is the projector onto the space $S_k$ of $k$ simplices in the complex.
The full boundary operator on the fully connected complex (the set of all subsets of $n$ points)
is the direct sum of the $k$-dimensional boundary operators, namely
$
  \partial =  \bigoplus_k \partial_k.
$

The \emph{$k$-homology group} is the quotient space $\mathbb{H}_k := \ker(\tpar_k)/ \img(\tpar_{k+1})$, representing all $k$-holes which are not ``filled-in'' by $k+1$ simplices and counted
once when connected by $k$ simplices (e.g., the two holes at the ends of a tunnel count once). Such global structures moulded by local relationships is
what is meant by the ``shape'' of data. 
The \emph{$k$th Betti Number} $\beta_k$ is the dimension of this $k$-homology group, namely
\[
	\beta_k := \dim \mathbb{H}_k \;\; .
\]
These Betti numbers therefore count the number of holes at scale $\veps$, as described above. By computing the Betti numbers at different scales $\veps$, we can obtain the \emph{persistence barcodes/diagrams}~\cite{ghrist2008barcodes}, i.e.,
a set of powerful interpretable topological features that account for different scales while being robust to small perturbations and
invariant to various data manipulations. These stable persistence diagrams not only provide information at multiple resolutions, but they also help identify, in an unsupervised fashion, the resolutions at which interesting structures exist.

\paragraph{Betti number estimation:}
The \emph{Combinatorial Laplacian}, or Hodge Laplacian, of a given complex is defined as $\Delta_k := \tpar_k^{\dagger}\tpar_k + 	\tpar_{k+1}\tpar_{k+1}^{\dagger}.$ From the Hodge theorem~\cite{friedman1998computing, lim_hodge_2019},
we can compute the $k$th Betti number as
\begin{equation}
	\beta_k := \dim \ker(\Delta_k).
\end{equation}
Therefore, computing Betti numbers for TDA can be viewed as a rank estimation problem (i.e., $\beta_k = \dim\tH_k - \rank(\Delta_k)$).

The problem of  Betti number estimation (BNE) can be defined as follows~\cite{gyurik2020towards}:
Given a set of $n$ points, its corresponding Vietoris–Rips complex $\Gamma$, an integer $0\leq k\leq n-1$, and the parameters $(\epsilon,\eta)\in (0,1)$, find the random value
$\chi_k\in[0,1]$ that satisfies with probability $1-\eta$ the condition
\begin{equation}
    \left| \chi_k - \frac{\beta_k}{\dim\tH_k}\right| \leq \epsilon,
\end{equation}
where $\dim\tH_k$ is the dimension of the Hilbert space spanned by the set of $k$-simplices in the complex (i.e., the number of $k$-simplices $|S_k|$ in $\Gamma$).
Gyurik et al.~\cite{gyurik2020towards} showed that this problem is likely classically intractable; specifically, it is shown that a generalization of the problem is as hard as simulating the one clean qubit model (DQC1-hard), which is believed to take super-polynomial time to compute on a classical computer. They also argued that BNE is likely to be immune to dequantization. In this paper, we discuss and address quantum algorithms for BNE.

\subsection{QTDA Algorithm}\label{ssec:qtda}
The seminal approach of \lloyd to estimate the Betti numbers using quantum computers, which was further analyzed by Gunn  and Kornerup~\cite{gunn2019review} and Gyurik et al.~\cite{gyurik2020towards}, comprises two main steps.
The first step of the algorithm is to create a mixed state $\rho_k$ over the states $\ket{s_k}$ of $k$-simplices (over $\tH_k$) that are in the  complex $\Gamma$.
The second step is to use Hamiltonian simulation (of the boundary operator or the Laplacian $\Delta_k$) and  quantum phase estimation (QPE) with $\rho_k$ in the input register (repeatedly projecting simplices from the complex onto the kernel) to estimate the kernel dimension of the Laplacian.

In order to prepare the maximally mixed state $\rho_k$ as part of the first step,
the QTDA algorithm first uses Grover's search algorithm~\cite{boyer1998tight} to construct the $k$-simplex state 
\[
\ket{\psi_k} = \frac{1}{\sqrt{|S_k|}}\sum_{s_k\in S_k}\ket{s_k},
\]
for the set $S_k$ with $|S_k| = \dim\tH_k$. Then, the mixed state 
\[
\rho_k =  \frac{1}{|S_k|}\sum_{s_k\in S_k}\ket{s_k}\bra{s_k}
\]
can be prepared from $\keti{k}$ by applying the CNOT gate to each qubit and tracing out into the ancilla zero qubits. The time complexity of this step is $O\left(\frac{k^2}{\sqrt{\zeta_k}}\right)$, where 
$\zeta_k := \frac{|S_k|}{{n \choose k+1}}$ is the fraction of $k$-simplices that are in the complex $\Gamma$. The number of gates 
required for this step is $O\left(kn^2+\frac{nk}{\sqrt{\zeta_k}}\right)$~\cite{gunn2019review}. We believe this step is unnecessary, because a random simplex of order $k$ can be drawn from the complex efficiently. Nevertheless, the same Grover's search is needed to restrict the boundary operator to the complex in the next step. 


The second step uses QPE to estimate the kernel dimension of the Laplacian $\Delta_k$. For this, the following \emph{Dirac operator} (the square root of the generalized Laplacian)
\begin{equation}
\tilde{B} = 
    \begin{pmatrix}
0 & \tpar_1 & 0 & \cdots & \cdots & 0\\
 \tpar_1^{\dagger} & 0 & \tpar_2 & 0 & \cdots & 0\\
0 &  \tpar_2^{\dagger} & 0 & \ddots  & \cdots & 0\\
\vdots & \vdots &\ddots &\ddots &\ddots &\vdots \\
\vdots &\vdots &\vdots &\ddots & 0  &\tpar_{n-1} \\
0 & 0 & 0 & \cdots & \tpar_{n-1}^{\dagger}& 0\\
\end{pmatrix}
\end{equation}
is first simulated such that  $\tilde{B^2} = \bdiag\left[\Delta_1, \ldots, \Delta_n\right]$  is a block diagonal matrix\footnote{The block diagonal form is obtained in the Hamming weight sorted representation of the simplices.}, since $\tpar_k\tpar_{k+1}= 0$. 
Given that $\tilde{B}$ has the same nullity (kernel) as $\tilde{B}^2$, the idea is to use Hamiltonian simulation of $\tilde{B}$ (i.e., implement $U = e^{i\tilde{B}}$), and use QPE with $\rho_k$ (computed in the first step) as the input state to estimate its eigenvalues. Since $\tilde{B}$ is an $n$-sparse Hermitian with entries $\{0,\pm1\}$, it is claimed that this can be simulated using $O(n)$ qubits and $O(n^2)$ gates~\cite{low2017optimal}.\footnote{The serious issue is that the restricted Dirac operator is not on hand, and requires $\tilde{P}_k$ to be known; see Remark \ref{rem:1}.}

QPE yields an approximate estimate of the eigenvalues of $\Delta_k$. We need to scale $\Delta_k$ such that its spectrum is in the interval $[0,1]$, in order to avoid multiples of $2\pi$; see Section~\ref{sec:analysis} for details on scaling. Supposing the smallest nonzero eigenvalue of (the scaled) $\Delta_k$ is greater than $\delta>0$, we then need to estimate the eigenvalues with a precision of at least $\frac{1}{\delta}$ in order to distinguish an estimated zero eigenvalue from others.
Therefore, the time complexity of this step is $O(\frac{n^2}{\delta})$ and requires as many gates for its implementation. 

This use of QPE provides us with an approximate estimate of some random eigenvalue of $\Delta_k$. For BNE with additive error $\eps$, we need to repeat the two steps $O(\eps^{-2})$ times. Hence, the total time complexity of QTDA (original \lloyd version\footnote{Except that we adjust the cost of QPE under our spectral interval assumption.}) for BNE with $ \left| \chi_k - \frac{\beta_k}{\dim\tH_k}\right| \leq \epsilon$ is given by
\[
O\left(\frac{n^4}{\epsilon^2\delta\sqrt{\zeta_k}}\right).
\]

\begin{remark}[Time Complexity Discrepancy]\label{rem:1}
We note that there is a discrepancy in the total time complexity of the QTDA algorithm reported in~\cite{lloyd2016quantum} and in the subsequent articles~\cite{gunn2019review,gyurik2020towards}, primarily due to differences in the underlying assumptions.
This relates to simulation of the matrix  $\tilde{B}$ or $\Delta_k$, where \lloyd suggest the requirement of constructing and applying the projector $\tilde{P}_k$ at each  round (possibly using Grover's search algorithm, although some implementation details are missing). Hence, the total time complexity in~\cite{lloyd2016quantum} is a product of the time complexity of the two steps. In contrast, the follow-up studies by Gunn  and Kornerup~\cite{gunn2019review} and Gyurik et al.~\cite{gyurik2020towards} assume that we have access to $\tilde{B}$ or $\Delta_k$ as an $n$-sparse matrix, in order to simulate it in the second step, and therefore the time complexities of the two main steps are added in~\cite{gunn2019review,gyurik2020towards} to obtain the total computational time complexity.
\end{remark} 

The subsequent articles by Gunn  and Kornerup~\cite{gunn2019review} and Gyurik et al.~\cite{gyurik2020towards} do not address the issue of efficient quantum construction of $\tilde{B}$ or $\Delta_k$ from the pairwise distances of the $n$ points, and assume that oracle access is given to the nonzero entries of $\tilde{B}$ and their locations.
The next section presents our approach that addresses all of these factors and achieves, for the first time, NISQ-QTDA.


\section{NISQ-QTDA}\label{sec:nisq}
We now present our algorithm for BNE with a near-term (NISQ) implementation. The algorithm introduces three key innovations, namely: (i) an efficient quantum representation of $\Delta_k$ as the sum of Pauli operators; (ii) a quantum rejection sampling approach to substitute Grover's search algorithm; and (iii) a stochastic rank estimation method to replace QPE. 

\subsection{Efficient Representation of $\Delta_k$}\label{ssec:Delrep}

As observed in the previous section, QTDA and BNE require us to simulate the Laplacian $\Delta_k$ (or the Dirac operator $\tilde{B}$). However,  the efficient quantum construction and simulation of $\Delta_k$ from just the pairwise distances of the given points, for all $k\leq n$, is not addressed by the previous work in~\cite{lloyd2016quantum,gunn2019review,gyurik2020towards}. 
We therefore present a quantum representation of the boundary operator and the combinatorial Laplacian that has a NISQ implementation. The representation involves only unitary operators and a sum of Pauli operators~\cite{van2020circuit}, and we show that only $n$ qubits, $O(n^2)$ gates, and an ${O}(n)$-depth circuit are required for its implementation. Note that such a (efficient, quantum) representation of $\Delta_k$ was left as an open problem in~\cite{gyurik2020towards}.

We begin by defining the matrix $a$ (familiar to some as the single spin annihilation operator from second quantization quantum mechanics) as
\begin{eqnarray*}
  a &= &\frac{(\sigma_x + i\sigma_y)}{2}
         \quad = \quad \begin{pmatrix}
           0&1\\
           0&0
           \end{pmatrix},
\end{eqnarray*}
where $a$ is written in terms of the Pauli operators $\sigma_x$ and $\sigma_y$. Note, by inspection or by recognizing
the coefficient of $\sigma_y$ is
imaginary, that both $a$ and the boundary operator below are not Hermitian. 
The full boundary map operator $\partial =  \bigoplus_k \partial_k$
of all possible simplices can now be written in the following manner
\begin{eqnarray*}
  \partial & = & a \otimes I \otimes I \otimes \ldots I \\
           & & + \sigma_z \otimes a \otimes I \otimes \ldots I \\
            & & + \sigma_z \otimes \sigma_z \otimes a \otimes \ldots I \\
            &  &\qquad \qquad  \vdots\\
  & &+ \sigma_z \otimes \sigma_z \otimes \sigma_z \otimes \ldots  \otimes a\\
  &= & \sum_{i=0}^{n-1} a_i ,
\end{eqnarray*}
where the $a_i$ are the Jordan-Wigner~\cite{jordan1928paulische} Pauli embeddings corresponding to the $n$-spin fermionic annihilation operators.

There are $n$ single qubit operations in the tensor product in each of the above terms, one for each vertex in the simplicial complex. The $a$ term ``removes a vertex'' from a simplex and the instances of $\sigma_z$ provide the 
correct
sign in the oriented sum by accumulating as many instances of $-1$ as there are vertices present to the left of the vertex being removed (also known as anti-symmetrizing the wave function). It would be interesting to explore other embeddings such as the Bravyi-Kitaev embedding~\cite{bravyi2002fermionic}.

Next, we transform this non-Hermitian operator into an associated Hermitian operator by adding its complex conjugate as
\begin{align*}
  B  &= \partial^\dagger + \partial = \sum_{i=0}^{n-1} a_i + a_i^\dagger \;\; .
\end{align*}

Since the $\sigma_y$ terms cancel
out,
the matrix $B$ will only have $\sigma_x$ and $\sigma_z$ terms. Since $B$ is expressed as a sum of polynomially-many Pauli terms~\cite{lloyd1996universal,whitfield2011simulation,childs2018toward,gui2020term,van2020circuit},
we can simulate $e^{iB}$ efficiently using the  Trotter-Suzuki  formula~\cite{lloyd1996universal,whitfield2011simulation,childs2018toward}. The depth complexity is usually $O(n^2)$ but, due to 
structure-exploiting
gate cancellations~\cite{van2020circuit}, we fortunately arrived at a depth complexity of ${O}(n)$; see Appendix~\ref{app:BOsim} for the details.

This form operates on the entire $n$-qubit space $\mathbb{C}^{2^n}\cong  \bigoplus_{k=0}^n \mathcal{H}_k$, and is not restricted to any subset of the simplices. If we are interested in the boundary map $ \partial_k:  \mathcal{H}_k \rightarrow {\mathcal{H}}_{k-1}$ operating on (all) $k$-dimensional simplices, then we need to apply  projections to  the operator $B$ in the following manner
\begin{align*}
  \partial_k  &= P_{k-1} B P_{k},
  \end{align*}
where $P_{k}$ is the projection onto all simplices of order $k$ (i.e., to the subspace $\mathcal{H}_k$).

We now define the combinatorial Laplacian corresponding to all the simplices in a given complex $\Gamma$ as
\begin{align}\label{eq:hodge}
  \Delta &= \PG B \PG B \PG,
\end{align}
where $\PG$ is the projector onto all simplices in $\Gamma$. We need to apply the  projection three times because $B$ includes the boundary conjugate, which must also be restricted to the simplices in the complex. Interestingly, if $\PG = I$, that is, if we have all possible simplices  in the complex $\Gamma$, then $ \Delta  = B^2  = nI$. Hence, the kernel is empty and there are no ``holes''.
To see that $B^2$ is proportional to the identity, it suffices to observe that the simplices produced by going up and then down in simplicial order are the opposite signs of those produced
by going down and then up, except for the simplex being acted upon which gets reproduced $n$ times.

Finally, the combinatorial Laplacian of $k$-simplices  $\Delta_k$ is then given by
\begin{equation}\label{eq:Delk}
   \Delta_k = P_k \Delta P_k, 
\end{equation}
where $P_k$ is a projector onto all $k$-simplices, and we know that the $k$-Betti number is given by $\beta_k = \dim\ker(\Delta_k)$.
We can simulate the Laplacian $\Delta_k$ efficiently on a quantum computer, because
the matrix $B$ can be simulated
effectively and the projectors $\PG$ and $P_k$ can be implemented with efficient unitary operators followed by measurement (see the next section).
We next describe how to construct the projectors $\PG$ and $P_k$.

\subsection{Projection onto Simplices}\label{ssec:projection}
The second key aspect of our approach consists of constructing the Vietoris-Rips simplicial complex $\Gamma$ from the pairwise distances of the points and  computing the projectors $\PG$ and $P_k$.
We now present a quantum sampling and control-based projection approach, which replaces the Grover's search algorithm in~\cite{lloyd2016quantum}, to project onto the simplices that exist in the complex $\Gamma$ and a separate projection onto simplices of order $k$.

\paragraph{Simplicial complex construction:}
\lloyd propose using Grover's search with an indicator function that returns one when all pairs of vertices of the given simplex are pairwise $\veps$-close. A QRAM version of this calculation was also suggested, which allows for quantum parallelism of the distance calculations and over multiple $\veps$ values. While Grover's search algorithm is probably provably optimal 
for complex projection\footnote{Using the generic theorem that Grover's search is optimal for unstructured search}, the proposed approach requires QRAM or at least fault-tolerance, both of which have not yet been achieved on currently available quantum devices. We therefore devise an alternative approach that is NISQ compatible.

To start, assume that we have access to all pairwise distances of the data-points (classically pre-calculated), followed by classical encoding of the $\varepsilon$-close pairs (known as the Vietoris-Rips 1-skeleton or the adjacency graph of the complex), where $\varepsilon$ is the resolution/grouping scale. We then systematically entangle the $k$-simplices with an $n/2$-qubit flag register. The $n/2$ qubits are used to process $n/2$ pairs of vertices at a time in $n-1$ rounds, thereby covering all $n\choose2$ potential $\varepsilon$-close pairs of vertices.

More precisely, the projection begins by creating a uniform superposition over all simplices (or over $k$-simplices, if the simplex order projection is executed first). Since the adjacency graph of the complex is already given, it remains only to keep a simplex for which every possible pair of vertices, for vertices of that simplex, is $\varepsilon$-close by checking the adjacency graph. We check $n/2$ pairs at a time, for all simplices in the superposition (hence the 
quantum speedup). The $n/2$ pairs are chosen such that the control gates may be executed in parallel. The actual gate is a C-C-NOT (Toffoli gate),
controlling the chosen pair of vertex qubits into the flag register. A particular gate for a pair is added to the circuit only if that pair of vertices is \emph{NOT} in the adjacency graph. For a given pair, if they are not $\varepsilon$-close, then all simplices containing this pair are not in the complex.\footnote{For a 1-skeleton with many more pairs of vertices that are in the adjacency matrix than those that are not in the complex, 
to save on Toffoli gates, we could alternatively code for simplices in rather than not in, but this approach needs to exclude the vertices absent from simplices, which can be achieved with a round of NOTs and CNOTs.}

We check for all $n\choose2$ pairs employing a cyclic shift approach with $n-1$ such rounds. In each round, we measure the flag register and proceed only if we receive all zeros. This collapses the simplex superposition into those simplices that have pairs which are not missing from the adjacency graph. Mid-circuit measurements allow us to reuse a common $n/2$-qubit flag register in each round. To store the results, we 
need $n-1$ measure-and-reset operations along with  $n\choose2$ classical registers.

Supposing $\zeta_k$ is the fraction of all possible simplices of order $k$ that are in the given complex $\Gamma$, then the collapse succeeds with probability $\frac{1}{\zeta_k}$. The order $k$ depends on the projector $P_k$ and the Laplacian $\Delta_k$ we consider. We therefore repeat this procedure $\frac{1}{\zeta_k}$ times.
While this procedure is quadratically less efficient
than Grover's search algorithm, which succeeds in identifying  all simplices in the complex when run for $\frac{1}{\sqrt{\zeta_k}}$ time,
our method does not require QRAM or fault-tolerance, and is NISQ compatible. The number of gates required for our procedure 
is $O(n^2 \bar{\zeta}_{k})$ with $\bar{\zeta}_{k} :=\min\{1-\zeta_k,\zeta_k\}$, whereas the depth of this circuit is only ${O}(n)$ since $n/2$ of these gates are in parallel. 


\paragraph{Projection onto $k$-simplices:}
Now, we describe the construction of the projector onto the $k$-simplex subspace. 
Our approach begins with the construction of a circuit that conditionally implements a count increment (i.e., $+1$). We condition on each qubit of the $n$-simplex register to increment a $\log(n)$-sized count register. This entangles the simplex register with the count register in such a way as to have each simplex entangled with the binary representation of its order. We then measure the count register and obtain a specific simplex order with some probability, collapsing the simplex register into a superposition of all simplices of that order only.
These steps involve an additional depth complexity of $O(\log^2 n)$, which maintains the $O(n)$-depth complexity for the circuit of our overall procedure.

\subsection{Stochastic Rank Estimation}
\label{ssec:chebyshev}

The final key ingredient of our proposed NISQ-QTDA algorithm is a stochastic rank estimation procedure that estimates the Betti number $\beta_k$ by estimating the rank of $\Delta_k$, which replaces the QPE component of the \lloyd algorithm.
The standard approach to estimate the rank of a square matrix is to compute all of its  
eigenvalues and count the number of nonzero eigenvalues, for which prior work~\cite{lloyd2016quantum} has employed QPE. In this paper, we propose a rank estimation procedure that does not require any decomposition of the corresponding matrix. 
In particular, our rank estimation approach is based on the classical \emph{stochastic Chebyshev} method~\cite{ubaru2016fast,ubaru2017fast}.
Namely, the proposed approach recasts the rank estimation problem to one of  estimating the trace of a certain (step) function of the matrix. The trace is then approximated by a stochastic trace estimator, where the step function is approximated 
by a Chebyshev polynomial approximation.

\paragraph{Stochastic trace estimator:}
Given a Hermitian matrix $A\in\RR^{N\times N}$, the stochastic trace estimation method~\cite{hutchinson1990stochastic,avron2011randomized} uses only the moments of the matrix to approximate the trace.  In the classical setting, ${\trace (A)}$ is estimated  by first generating random vector states $|v_l\rangle$  with random independent and identically distributed (i.i.d.) entries, $l=1,..,\nv$,
and then computing the average over the moments $\langle v_l | A | v_l\rangle $; namely,
\begin{eqnarray}
\label{eq:trace_estimator}
\trace (A)  \approx \frac{1}{\nv}  \sum_{l=1}^{\nv} \langle v_l | A | v_l\rangle.
\end{eqnarray}
Any random vectors $|v_l\rangle$ with zero mean and uncorrelated coordinates can be used~\cite{avron2011randomized}.

In the quantum setting, however, particularly with NISQ computations, generating random states $|v_l\rangle$ of exponential size with i.i.d.\ entries is not viable. 
Alternatively, it has  been shown that random columns drawn from the Hadamard matrix work very well in practice for stochastic  trace estimation~\cite{fika2017stochastic}.
Sampling a random Hadamard state vector in a quantum computer is extremely simple and can be conducted with short-depth circuit. Given an initial state $\ket{0}$, we randomly flip the $n$ qubits (possibly by applying a NOT gate as determined by a random $n$-bit binary number $\in [0,2^n-1]$ generated classically). Thereafter, we simply apply Hadamard gates to all qubits. This produces a state corresponding to a random column of the $2^n \times 2^n$ Hadamard matrix.

The columns of a Hadamard matrix have pairwise independent entries.
Hence, we consider the random state vector $| v_l\rangle = | h_{c(l)}\rangle$, i.e., some random Hadamard column with $c(l)$ defining the random index, and then we estimate the moments
$\langle h_{c(l)} | A | h_{c(l)}\rangle$ and average over the $\nv$ samples to approximate the trace. The error analysis for this approach is presented in Appendix~\ref{app:proof}.

Alternatively, quantum t-design circuits are a popular way to generate pseudo-random states~\cite{ambainis2007quantum,ji2018pseudorandom,brakerski2019pseudo}. A t-design circuit outputs a state that is indistinguishable from states drawn from a random Haar measure. These t-designs in a quantum computer are equivalent to $t$-wise independent vectors in the classical world~\cite{ambainis2007quantum}. Short-depth circuits exist (though not as short as above) that are approximate $t$-designs~\cite{brakerski2019pseudo}. Such $t$-design circuits can be used to generate the random states $\ket{v_l}$  for  trace estimation. Indeed, random vectors with just $4$-wise independent entries suffice for trace estimation (we omit the details here because this approach is less competitive than the above super-short-depth Hadamard construction).

\paragraph{Chebyshev approximation:}
Assuming the smallest nonzero eigenvalue of $A$ is greater than or equal to $\delta$, then the rank of $A$ can be written as
\begin{equation}
\label{eq:hoft} 
\rank(A) \defeq \trace(h(A)), \ \mbox{where} \ h(x) = \left\{\begin{array}{l l}
1  & \ \textrm{if}\ \ x \ > \delta\\
0  & \ \textrm{otherwise}\\
\end{array} \;. \right.
\end{equation}
Given the eigen-decomposition $A =\sum_{i}\lambda_i|u_i\rangle\langle u_i|$, we have the matrix function $h(A)= \sum_{i} h(\lambda_i)|u_i\rangle\langle u_i|$ where the step function $h(\cdot)$ takes a value of $1$ above the threshold $\delta>0$. The parameter $\delta$ is assumed to be known  (or, in the classical setting, can be estimated using the spectral density method~\cite{ubaru2016fast}).
In the case of TDA,  for many  simplicial-complex types,  a lower bound for the smallest nonzero eigenvalue of $\Delta_k$ can be estimated; refer to Section~\ref{ssec:Qadv} for a few examples.

Next, the approach of Ubaru et al.~\cite{ubaru2016fast,ubaru2017fast} consists of approximating the matrix function $h(A)$ by employing Chebyshev polynomials~\cite{trefethen2019approximation}, and estimating the trace using the stochastic
estimator~\eqref{eq:trace_estimator}.
More specifically, $h(A)$ is approximately expanded in the following manner
\begin{equation}\label{eq:Chebyshev}
 h(A) \approx \sum_{j=0}^mc_jT_j(A),
\end{equation}
where  $T_j$ is the $j$th-degree Chebyshev polynomial
of  the first kind and $c_j$ are the coefficients; see Appendix~\ref{app:chebyshev} for further details on Chebyshev polynomials and coefficient computation.
Therefore, the rank of a given matrix $A$, with the smallest nonzero eigenvalue greater than or equal to $\delta$, can be approximately estimated using the stochastic Chebyshev method~\cite{ubaru2016fast,ubaru2017fast} as
\begin{equation}\label{eq:rank} 
\rank(A)\approx
\frac{1}{\nv}\sum_{l=1}^{\nv}\left[\sum_{j=0}^m c_j\langle v_l|T_j(A)|v_l\rangle\right].
\end{equation}

The method estimates the rank using only the Chebyshev moments of the matrix $\langle v_l|T_j(A)|v_l\rangle$.
Classically, these moments are typically built using the three-term recurrence~\cite{ubaru2016fast}. Although we cannot employ the recurrence approach with quantum computers, we can compute the moments $\bra{v_l}A^i \ket{v_l}$ for a given order $i$. We therefore  use a general (summation) formula given by
 \begin{equation}\label{eq:Cheb_sum}
 T_j(x) = \sum_{i = 0}^{\lfloor \frac{j}{2}\rfloor}(-1)^i 2^{j-(2i+1)} g(j,i) x^{j-2i},
 \end{equation}
where 
\[
g(j,i) = \frac{\binom{2i}{i}\binom{j}{2i}}{\binom{j-1}{i}},
\]
 to compute the Chebyshev moments $\langle v_l|T_j(A)|v_l\rangle$ using the moments $\{ \bra{v_l}A^i \ket{v_l}\}_{i=0}^j$. Details on computing these moments for the Laplacian $\Delta_k$ are given in Appendix~\ref{app:moments}.

We now have all the ingredients to present our NISQ-QTDA algorithm as follows.

 \begin{algorithm}[H]
\caption{NISQ-QTDA Algorithm}
\label{alg:algo1}
\begin{algorithmic}
   \STATE {\bfseries Input:} Pairwise  distances of $n$ data points and encoding of the $\veps$-close pairs; parameters $\eps,\delta$, and $\nv=O(\eps^{-2})$; and $\nv$ $n$-bit random binary numbers.
 \STATE {\bfseries Output:} Betti number estimates $\chi_k, \: k=0,\ldots, n-1$.
 \FOR{$l=1,\ldots, \nv=O(\eps^{-2})$}
 \FOR{$i=0,\ldots,m= O(\log(1/(\delta\eps)))$}
 \STATE  {\bfseries  1.} Prepare a random Hadamard  state vector $\ket{v_l}$ from $\ket{0}$ using the $l$-th random number.
\STATE  {\bfseries  2.} Apply $P_k$, apply $\PG$, and simulate $B$ to obtain \\ \quad~$\ket{\phi_l^{(i)}} = \left(\prod_{j=0}^{j=i-1} \PG P_k^{j  \textrm{\%} 2} B\right) \PG P_k \ket{v_l}$;\footnotemark[11] \\ \quad~Save the value of $\dim \tH_k$ the first time by leaving out $\PG P_k^{j  \textrm{\%} 2} B$ (i.e., $i=0$).
\STATE  {\bfseries  3.} Compute the moments $\mu_l^{(i)} = \bra{v_l}\Delta_k^i \ket{v_l} $ as $ \mu_l^{(i)}= \left\langle\phi_l^{(i)}\bigg|\phi_l^{(i)}\right\rangle$  or $= \left\|\ket{\phi_l^{(i)}}\right\|^2$; \\
\quad~Readout the $\mu_l^{(i)}$ and compute the Chebyshev moments $\theta_l^{(j)} = \bra{v_l}T_j(\Delta_k) \ket{v_l} $ using~\eqref{eq:Cheb_sum}.
\ENDFOR
\ENDFOR
\STATE Estimate $\chi_k = 1-\frac{1}{\nv\dim \tH_k}\sum_{l=1}^{\nv}\left[\sum_{j=0}^m c_j\theta_l^{(j)}\right]$.
\STATE \emph{Repeat} for $k=0,\ldots, n-1$.
\end{algorithmic}
\end{algorithm} 
\footnotetext[11]{The modulo operator is denoted by $\%$, such that $\% 2$ has the effect of alternating between excluding and including $P_k$. The state is written in non-normalized form, which is equivalent to collecting statistics through post-selection.}

The next section presents our analysis of the error and the computational complexities of the above NISQ-QTDA algorithm.

\section{Theoretical Analyses}\label{sec:analysis}
We turn to the theoretical analysis of our proposed NISQ-QTDA algorithm, first presenting an error analysis that provides bounds on the number of random vectors $\nv$ and the polynomial degree $m$ needed to achieve BNE  with $ \left| \chi_k - \frac{\beta_k}{\dim\tH_k}\right| \leq \epsilon$, and subsequently presenting the gate and time complexities of the algorithm.
We then discuss different scenarios under which the QTDA algorithms can achieve significant speedups over classical algorithms, including when our proposed algorithm can be NISQ implementable.

\subsection{Error Analysis}\label{ssec:Err}

Algorithm~\ref{alg:algo1} returns a Betti number estimate $\chi_k$ for each order $k=0,\ldots,n-1$. We show that, for the appropriate choice of $m$ and $\nv$, this
estimate is a BNE with an additive error $\eps\in(0,1)$. Our main result is presented 
in the following theorem.

\begin{theorem}\label{theo:main}
Assume we are given the pairwise distances of any $n$ data points and the encoding of the corresponding $\veps$-close pairs, together with an integer $0\leq k\leq n-1$ and the parameters $(\epsilon,\delta,\eta)\in (0,1)$.
Further assume the smallest nonzero eigenvalue of the scaled Laplacian $\tilde{\Delta}_k$ is greater than or equal to $\delta$, and choose $\nv$ and $m$ such that
\[
\nv = O\left(\frac{\log(2/\eta)}{\epsilon^2}\right)
\qquad\qquad \mbox{and} \qquad\qquad
m \geq 
 \frac{\log\left(\frac{32\log(2/\epsilon)}{\pi\delta\epsilon}\right)}{\log\left(1+\frac{\pi\delta}{4\log(2/\epsilon)}\right)}.
\]
Then, the Betti number estimation $\chi_k \in[0,1]$ by Algorithm~\ref{alg:algo1} satisfies
 \[
  \left| \chi_k - \frac{\beta_k}{\dim\tH_k}\right| \leq \epsilon,
  \]
with probability at least $1-\eta$ .
\end{theorem}

The proof of Theorem~\ref{theo:main} is given in Appendix~\ref{app:proof}. We assume the Laplacian $\Delta_k$ has eigenvalues in the interval $[0,1]$. Since the largest eigenvalue is bounded by $O(k)$~\cite{gunn2019review}, we can scale $\Delta_k$ by $O(1/k)$ to ensure the spectrum of the scaled $\tilde{\Delta}_k$ is in $[0,1]$. Note that QPE also requires this scaling. 
Next, since the step function is a discontinuous function, we consider
in the proof of the above theorem an analytic surrogate function that approximates the step function, and then use a Chebyshev approximation of this surrogate function. In particular, we consider that the hyperbolic tangent function
$f(x) = \frac{1}{2}\left(1+\tanh(\alpha(x - \frac{\delta}{2}))\right)$, with $\alpha  = \frac{1}{\delta}\log(\frac{2}{\epsilon})$, approximates the step function well in the region of interest. We assume that the smallest nonzero eigenvalue of the (scaled) combinatorial Laplacian $\Delta_k$ is greater than or equal to $\delta$, as assumed in the theorem.
The parameter $\delta$ is assumed to be given, as previously noted, and can be estimated for many types of simplicial-complexes; see the discussion in the next section.
The above theorem also accounts for the errors in the estimation of the moments $\mu_l^{(i)} = \bra{v_l}\Delta_k^i \ket{v_l} $ using the quantum computer, i.e., the errors in the simulation of $e^{iBt}$ and moment estimation. These errors are additive in nature and will be small (see the Appendix for details), noting that a choice of small $t\ll 1$ suffices, and the errors become negligible when we scale the estimates by $\frac{1}{\dim\tH_k}$.

\subsection{Complexity Analysis}\label{ssec:complexity}
We now discuss the circuit and computational complexities of our proposed algorithm and show that it is NISQ implementable under certain conditions, such as clique-dense complexes which commonly occur for large resolution scale and high order $k$. 
The main quantum component of our algorithm comprises the computation of the moments $\mu_l^{(i)} = \bra{v_l}\Delta_k^i \ket{v_l}$, for $i=0,\ldots,m= O(\log(1/(\delta\eps)))$, and the computation of the $\nv=O(\eps^{-2})$ random Hadamard vectors. The random Hadamard state preparation requires $n$ single-qubit Hadamard gates in parallel and $O(1)$ time.

Next, for a given $k$, constructing the combinatorial Laplacian $\Delta_k$ involves simulating the boundary operator $B$ and constructing the projectors  $\PG$ and $P_k$.
The operator $B$, involving the sum of $n$ Pauli operators, can be simulated as $e^{iB}$ using a circuit with $1$ ancillary qubit and  $O(n)$ gates. Hence, the time complexity will also be $O(n)$. For the projectors, constructing $P_k$ requires $O(n\log^2n)$ gates, and this succeeds for a random order $k$. Then for $\PG$, we need to  find all the simplices that are in the complex $\Gamma$. This is achieved using $n/2$ qubits in parallel and  $n-1$ operations, and thus the time complexity remains $O(n)$. The number of gates required will be $O(n^2 \bar{\zeta}_{k})$, recalling $\bar{\zeta}_{k} :=\min\{1-\zeta_k,\zeta_k\}$.  The procedure of applying the projectors succeeds with probability $1/\zeta_k$, and we obtain the projection onto all simplicies of order $k$ that is in $\Gamma$, for a random $k$.
Hence, the projectors together require $O(n^2)$ gates, while the depth remains $O(n)$, and the time complexity for a projection will be $O\left(\frac{n}{\zeta_k}\right)$, since  $P_\Gamma$ succeeds when run $O(\zeta_k^{-1})$ times.

For higher moments, we need to construct  $\Delta_k^i$ up to the power $m= O(\log(1/(\delta\eps)))$. Therefore, the circuit has a total gate complexity of 
$O(n^2\log(1/(\delta\eps)))$ with a depth of $O(n\log(1/(\delta\eps)))$.
In order to compute $\ket{\phi_l^{(i)}}$ for a given degree $(i)$, we need all $2 \cdot (i)$ applications of $\PG$ to succeed simultaneously, and therefore we need to run the projections $O(\zeta_k^{-2(i)})$ times. The first projection $P_k$ yields a random order $k$, and the subsequent projections onto a simplicial order needed for higher moments will also have to be onto the same order $k$. Due to the application of the boundary operator $B$, the subsequent simplicial order projections will result in a projection onto one of the simplicial orders $k-1$ or $k+1$ (after one application) and $k-2$, $k$ or $k+2$ (after two applications). Hence, we need to repeatedly apply the order projection (a constant number of times) in order to ensure that we are operating on the right order (in addition to the complex projection).
The procedure of computing the $m$ moments is repeated $\nv=O(\eps^{-2})$ times with different random Hadamard column vectors, and thus the total time complexity of our algorithm to compute the BNE $\chi_k$ is given by
\[
O\left(\frac{n\log(1/(\delta\epsilon))}{\epsilon^2\zeta_k^{2\log(1/(\delta\epsilon))}}\right).
\]

Note that we consider above a simple approach of repeating the projection until it succeeds each time, since this requires a short-depth circuit and results in a $\zeta_k^{-2\log(1/(\delta\epsilon))}$ term in the time  complexity.
Perhaps, this procedure can be improved to reduce the complexity further. We consider clique-dense complexes (i.e., large $\zeta_k$, as discussed in the next section), and this term is a constant for a given set $\{\eps,\delta,\zeta_k\}$ and is independent of $n$.

\subsection{Quantum Advantage}\label{ssec:Qadv}
Table~\ref{tab:comp} summarizes the circuit and computational complexities of our algorithm and compares them to that for the QTDA algorithm of~Lloyd et al.~\cite{lloyd2016quantum}. As remarked earlier, the gate and time complexities for this QTDA algorithm reported in Gyurik et al.~\cite{gyurik2020towards} and Gunn \& Kornerup~\cite{gunn2019review} are different from those reported in~Lloyd et al.~\cite{lloyd2016quantum}, since Gyurik et al.\ and Gunn \& Kornerup both assume the operator $\tilde{B}$ is given and thus they add the complexities of the two steps (Grover's algorithm and QPE); refer to Remark~\ref{rem:1} above.

\begin{table}[htb]

 \caption{Comparisons of the circuit and computational complexities for QTDA to compute BNE with an $\epsilon$ error, a $\zeta_k$ fraction of order-$k$ simplices in the complex, and a $\delta$ smallest nonzero eigenvalue of 
$\tilde{\Delta}_k$.}\label{tab:comp}
 \begin{center} 
 \begin{tabular}{|l|c|c|c|c|}
\hline
Methods & \# Qubits & \# Gates & Depth & Time \\
\hline
\lloyd & $2n+\log n+\frac{1}{\delta}$&  $O\left(\frac{n^2}{\delta\sqrt{\zeta_k}}\right)$ & $O\left(\frac{n^2}{\delta\sqrt{\zeta_k}}\right)$ & $O\left(\frac{n^4}{\epsilon^2\delta\sqrt{\zeta_k}}\right)$\\
Ours & $n+\log(n)$& $O(n^2\log(1/(\delta\epsilon)))$& $O(n\log(1/(\delta\epsilon)))$&$O\left(\frac{n\log(1/(\delta\epsilon))}{\epsilon^2\zeta_k^{2\log(1/(\delta\epsilon))}}\right)$ \\
\hline
\end{tabular}
 \end{center} 
\end{table}

For comparison, note that the best-known classical algorithm for BNE of order $k$ has a time complexity of $O({n \choose k})^2$~\cite{lloyd2016quantum} or $O(\mathrm{poly}(n^k))$~\cite{gyurik2020towards}.
Therefore, the QTDA algorithms can achieve exponential speedups over the best-known classical algorithms whenever we have: 
\begin{itemize}
    \item {\bf Simplices/Clique dense complexes}~--~the given complex $\Gamma$ is simplices/clique dense, i.e., $\zeta_k$ is large or $|S_k|\in O(\mathrm{poly} (n))$; and 
    \item  {\bf Large spectral gap}~--~the spectral gap between zero and nonzero eigenvalues of $\Delta_k$ is large, i.e., $\delta$ of $\tilde{\Delta}_k$ is not too small.
\end{itemize}
More importantly, our proposed Algorithm~\ref{alg:algo1} is \emph{NISQ} implementable whenever the Laplacian spectral gap is large, representing the only algorithm that is able to do so.

\paragraph{Simplices/Clique dense complexes:}
We first discuss examples of complexes that are simplices/clique dense. 
Gyurik et al.~\cite{gyurik2020towards} presented a few examples of a family of graphs that are clique-dense. Using the clique-density theorem~\cite{reiher2016clique}, we can consider a class of graphs/complexes that are clique-dense. Let $\gamma> \frac{k-2}{2(k-1)}$ be a constant.
Then, for a graph with $n$ nodes and $\gamma n^2$ edges and for a given order $k \geq 3$, we have $|S_k| = \Omega(n^{k+1})$ by the clique-density theorem~\cite{reiher2016clique}. If $\gamma \geq \frac{k-1}{k}$, then the graph will be even denser. Such clique-dense complexes occur in TDA when the resolution scale $\veps$ is large (close to maximum distance between points), and therefore QTDA algorithms can achieve a significant speedup for BNE over classical algorithms, particularly when we are interested in larger (and many) orders of $k$. We also refer to the discussions in~\cite{lloyd2016quantum, gyurik2020towards} on when quantum TDA algorithms are advantageous.

\paragraph{Laplacian spectral gap:}
We next discuss different settings, namely when the Laplacian of a given simplicial complex has a sufficiently large spectral gap such that a small degree $m$ will suffice for BNE.
Not much is known for general simplicial complexes in terms of lower bounds for $\delta$, the smallest nonzero eigenvalue of the combinatorial Laplacians~\cite{gyurik2020towards}. However, we can identify many specific examples of simplicial complexes for which  $\delta$ can be large. Indeed, several articles~\cite{goldberg2002combinatorial,horak2013spectra, yamada2019spectrum, lew2020spectral, lew2020spectral2} have studied  the  spectra of the Laplacian of different simplicial complexes, including random simplicial complexes~\cite{gundert2016eigenvalues,kahle2016random,knowles2017eigenvalue,adhikari2020spectrum,beit2020spectral}.

{\it Some specific complexes:} First, let us consider a few specific types of simplicial complexes. The articles by Horak and Jost~\cite{horak2013spectra} and Yamada~\cite{yamada2019spectrum} consider the Laplacian spectra of $k$-regular complexes and orientable complexes. A simplicial complex $\Gamma$ is $k$-regular if and only if all of its $k$-faces have the same degree $d_k$, whereas a $k+1$-dimensional simplicial complex $\Gamma$ is \emph{orientable} if and only if all $k$-faces of $\Gamma$ have orientation such that any two simplices which intersect on a $(k-1)$-face induce a different orientation on that face. For $k$-regular simplicial complexes with degree $d_k=1$, the Laplacian $\Delta_{k}$ has all nonzero eigenvalues equal to $k+2$. 
Horak and Jost~\cite{horak2013spectra} show similar results for higher degree and for orientable $k$-dimensional simplicial complexes.
Yamada~\cite{ yamada2019spectrum} presents lower bounds on the nonzero eigenvalues of the Laplacian for these two types of complexes in terms of the Ricci curvature~\cite{bauer2011ollivier} of the complex. For an orientable $k$-dimensional  simplicial complex $\Gamma$ with maximum degree $d_k$ for the $(k-1)$-faces, the smallest nonzero eigenvalues of $\Delta_k$, denoted by $\delta_k$, satisfies
\[
\delta_k \geq (k+1)(\kappa_c -1)+\frac{2}{d_k},
\]
where $\kappa_c$ is the Ricci curvature on $\Gamma$. If the complex is orientable $k$-regular, then the minimal eigenvalues of $\Delta_k$ satisfies $\delta_k \geq (k+1)\kappa_c$. We refer to~\cite{yamada2019spectrum} for bounds on the Ricci curvature $\kappa_c$ for $k$-regular complexes. 
Such complexes therefore can have a large spectral gap between zero and nonzero eigenvalues (i.e., large $\delta$ for the scaled Laplacian) when $k$ is sufficiently large.

Next, the article by Goldberg~\cite{goldberg2002combinatorial} considers the Laplacian spectra of a few specific complexes.  
For a finite simplicial complex $\Gamma$ that contains distinct flapoid clusters of size $d_c$, the nonzero eigenvalues of the Laplacian are all equal to $d_c = o(n)$. The article by Lew~\cite{lew2020spectral2} presents a lower bound for the spectral gap of the $k$-Laplacian $\Delta_k$ for complexes without missing faces.
In particular, for an $n$-vertex simplicial complex $\Gamma$ without missing faces of dimension larger than $\ell$, the smallest nonzero eigenvalue (spectral gap) of $\Delta_k$, for $k\leq \ell$, satisfies
\[
\delta_k \geq (\ell+1)(d_k +k+1)- \ell n,
\]
where $d_k$ is the minimal degree of a $k$-simplex in $\Gamma$. 
These complexes therefore can also have a large spectral gap, under appropriate conditions. 

{\it Random complexes:} Let us now consider random simplicial complexes. For a random complex $\Gamma$ with $n$ vertices and constants $C_1,C_2$ and $p\geq (k+C_1)\log(n)/n$, Gundert and Wagner~\cite{gundert2016eigenvalues} show that, if the expected degree of $k-1$ faces is $d_{k}:=p(n-k)$, then the normalized Laplacian\footnote[12]{The normalized Laplacian is defined as $\hat{\Delta}_{k} := D_k^{-1}\Delta_k$, where $D_k$ is the diagonal matrix with the degrees of the faces as its diagonal entries.} $\hat{\Delta}_{k}$ has all its nonzero eigenvalues in the interval 
\[
\left[1-\frac{C_2}{\sqrt{d_{k}}},1+\frac{C_2}{\sqrt{d_{k}}}\right],
\]
with high probability.
It was recently shown by Adhikari et al.~\cite{adhikari2020spectrum} that, for random dense graphs/complexes, the limiting spectral gap (between zero and nonzero eigenvalues) of the normalized Laplacian approaches $1/2$.
Another interesting and relevant result related to the spectra of random complexes was obtained by Beit-Aharon and Meshulam~\cite{beit2020spectral}, who consider random subset complexes.
Suppose $\Gamma$ is a full complex (also called a homological sphere) with all possible simplices of order up to $n-1$, i.e., an $n$-simplex. 
Let $\tilde{G}$ be a random subset of $\Gamma$, $\tilde{G}\subset \Gamma$, of size $\tilde{n}$ and let $\delta_k$ be the minimal (smallest nonzero) eigenvalue of the $k$-Laplacian $\Delta_k$ of $\tilde{G}$ for $k<n$. Then, for $k\geq 1$ and $\xi>0$, if the size $\tilde{n} = \lceil \frac{4k^2\log n}{\xi^2} \rceil$, we have~\cite{beit2020spectral} 
\[
\pr\left[  \delta_k < (1-\xi) \tilde{n}\right] \leq O\left(\frac{1}{n}\right).
\]
These results therefore suggest that random dense complexes will likely  have a large spectral gap between zero and nonzero eigenvalues. Indeed, this is exactly the regime (large $\zeta_k$) where the quantum algorithms are advantageous. As discussed by Lloyd~\cite{lloyd1996universal}, such dense complexes occur in TDA when the resolution scale $\veps$ is large. For such complexes, our proposed QTDA algorithm has great prospects to be NISQ implementable.

\paragraph{Approximate BNE:} When the spectral gap of the Laplacian $\tilde{\Delta}_k$ is not larger than the chosen threshold $\delta$, our NISQ-QTDA algorithm estimates an approximate Betti number by counting the (larger) eigenvalues above the threshold $\delta$. This was defined in~\cite{gyurik2020towards} as the problem of approximate Betti number estimation (ABNE). Such ABNE will be useful in certain situations, since our method provides an approach to filter out small (noisy) eigenvalues and only consider larger (dominant) eigenvalues for estimating the Betti number. These small nonzero eigenvalues occur when there are thinly (loosely) connected components in the complex.
Such connections likely occur when the resolution scale $\veps$ is small, and they might not persist when $\veps$ increases. Our approach therefore provides a way to filter out noise and  estimate the features that persist at larger resolution scale.
Moreover, we note that computing the moments of the Laplacians $\Delta_k$ (exponential in size) for different $k$ is non-trivial, and these moments can be used as features for certain downstream learning tasks, for example.
\section{Conclusion}
In this paper, we presented a new quantum algorithm for Betti number estimation in topological data analysis. For many types of simplicial complexes, our algorithm may be \emph{implemented on real quantum computers} that exist today, i.e., NISQ implementable, and achieves guaranteed exponential speedups over classical algorithms under widely-held assumptions. The stochastic Chebyshev method presented here opens the door for the development of new quantum algorithms for other closely related problems that occur in applications of numerical linear algebra~\cite{ubaru2016fast,di2016efficient}, computational physics~\cite{lin2016approximating} and machine learning~\cite{han2017approximating,ubaru2017fast2}, among others. To the best of our knowledge, the presented methodology is the \emph{first}
QML algorithm with $O(n)$-depth implementation and likely exponential speedup, which offers the potential for
our
methods to become the first set of useful algorithms to achieve quantum advantage on arbitrary input. 
In fact, we have implemented our full algorithm and successfully executed it on a real quantum computer. The corresponding experiments (whose results will be reported in a future paper) confirm and support the theoretical results presented herein.
Moreover, our techniques in this paper have wider potential applications including the problem of estimating numerical rank~\cite{ubaru2016fast}, spectral densities~\cite{lin2016approximating} and other spectral sums~\cite{di2016efficient,han2017approximating,ubaru2017fast2}.

\subsubsection*{Acknowledgement}
This research was supported in part by the Air Force Research Laboratory (AFRL) grant number FA8750-C-18-0098, and in part by IBM Research, South Africa under the Equity Equivalent Investment Programme (EEIP) of the government of South Africa.
We would like to acknowledge Tal Kachman for the suggestion to use controlled-increment to entangle the simplices with the count register, and Yang-Hui He, Vishnu Jejjala and Kugendran Naidoo for discussions on the boundary operator representation.
The authors would also like to thank Scott Aaronson, Paul Alsing, Aram Harrow and Vasileios Kalantzis
for valuable discussions.
We further thank the support of Tom Ether, Maletsabisa Molapo, Bob Wisnieff, and our management at IBM Research.


\newpage
\appendix
\section{Methodology Details}
This appendix provides further technical details related to different aspects of our proposed NISQ-QTDA method. We first discuss in more detail the simulation of the boundary operator $B$ and show how this can be achieved with a circuit comprising only $O(n)$ gates. We then present some additional details related to the Chebyshev polynomial method and moment computation in the context of a quantum computer. Appendix~\ref{app:proof} subsequently presents a proof of our main Theorem~\ref{theo:main}.

\subsection{Boundary Operator Simulation}\label{app:BOsim}
In the main body of the paper, we presented an efficient representation for the boundary operator $B$ as a sum of Pauli operators. Here, we present the details on the simulation of $B$. For a quantum implementation, we need to construct the unitary $U_B(t) = e^{-iBt}$ for some $t$. To do so, we begin with the Trotter-Suzuki  formula~\cite{lloyd1996universal}:
Supposing an $n$-qubit Hamiltonian is written as $A = \sum_{j=1}^n A_j$, we then have 
\[
e^{-i\sum_{j=1}^n A_j t} =\prod_{j=1}^n e^{-iA_j t} + O(n^2 t^2),
\]
where the error in this approximation is negligible when $t\ll 1$. Higher order variants of this formula also exist~\cite{childs2018toward}.
Note that an error is incurred since the $A_j$ terms need not commute (e.g., $A_j A_k \ne A_k A_j$) in general. However, in the case of the boundary operators $\partial$ and $B$ described in Section~\ref{ssec:Delrep}, each $A_j$ are Pauli terms $\{\sigma_x,\sigma_y,\sigma_z,\sigma_i\}^{\otimes n}$. Recent articles~\cite{van2020circuit,gui2020term} discuss the quantum simulation of a sum of Pauli terms, and show how certain Pauli terms commute with each other. Hence, by permuting blocks to align the terms that commute, we can cancel several gates in order to reduce the size (both gate and depth complexity) of the circuit. 

The Pauli terms can be mapped to quantum circuits  using diagonalization and unitary evolution~\cite{whitfield2011simulation}. The Pauli terms $\{\sigma_z,\sigma_i\}^{\otimes n}$ are already diagonal, and thus we have $e^{i\sigma_i t} = e^{it}I$ and
\[
  e^{i\sigma_zt} = 
  \begin{bmatrix}
  e^{it} & 0 \\
  0 &  e^{-it}\\
  \end{bmatrix}
  =: R_z(t).
\]
Using diagonalization, we have the Pauli $\sigma_x = H \sigma_z H^T$, where $H = \frac{1}{\sqrt{2}}   \begin{bmatrix}
  1 & 1 \\
  1 &  -1\\
  \end{bmatrix} $ is the Hadamard matrix, and we further obtain 
$e^{-i\sigma_x t} = H e^{i\sigma_zt} H^T  =   H R_z(t) H^T $~\cite{van2020circuit}. Next, for the boundary operator $B=\partial+\partial^{\dagger}$, we have
\begin{eqnarray*}
  B & = & \sigma_x\otimes I \otimes I \otimes \ldots I \\
           & & + \sigma_z \otimes \sigma_x \otimes I \otimes \ldots I \\
            &  &\qquad \qquad  \vdots\\
  & &+ \sigma_z \otimes \sigma_z \otimes \sigma_z \otimes \ldots  \otimes \sigma_x ,
\end{eqnarray*}
since $a+a^{\dagger} = \sigma_x$. Therefore, we can simulate $e^{iBt}$ using a quantum circuit with a few CNOT, Hadamard $H$, and rotation $R_z(t)$ gates. 

For example, supposing $n=4$, the basic circuit (with one ancillary qubit) we obtain for $B$ using the above diagonalization is given in Figure~\ref{fig:1}. Note that the circuit has $O(n^2)$ gates (i.e., $n(n+1)$ CNOT gates, $n$ rotation $R_z(t)$ gates, and $2n$ Hadamard $H$ gates) with depth $O(n^2)$. This circuit is currently not NISQ. We observe, however, that the circuit can be simplified significantly by identifying sub-terms that commute and then using gate cancellations. Indeed, since the sub-terms 
$\{\sigma_z\}^{\otimes r}$ for different $r$ values between the terms in $B$ commute, a number of CNOT gates can be cancelled (because they are redundant). Similar gate count reduction is explored in~\cite{van2020circuit,gui2020term} for other problems involving Pauli terms in Hamitonian simulation.

\begin{figure}[h]
\begin{equation*}
\Qcircuit @C=0.5em @R=0.0em @!R {
	 \lstick{q_{0}} &  \gate{H} & \ctrl{4} & \qw  & \ctrl{4} & \gate{H}    
	 & \ctrl{4} & \qw & \qw & \qw &  \qw  &  \qw & \ctrl{4}  & \qw                
	 & \ctrl{4} & \qw & \qw & \qw &  \qw  &  \qw &  \qw  &  \qw & \qw &\qw & \ctrl{4} 
	 & \qw &  \ctrl{4} & \qw  &\qw & \qw & \qw &  \qw  &  \qw &  \qw   &  \qw &  \qw &  \qw & \qw &\qw &\qw  & \ctrl{4}&\qw \\
	 \lstick{q_{1}} &  \qw & \qw & \qw & \qw &  \qw 
	 & \qw &  \gate{H} & \ctrl{3} & \qw  & \ctrl{3} & \gate{H} &\qw &\qw
	 & \qw & \qw & \ctrl{3} & \qw & \qw & \qw &  \qw  & \qw   & \ctrl{3} & \qw & \qw & \qw  & \qw & \qw &\ctrl{3} & \qw & \qw & \qw &  \qw  & \qw &\qw  &  \qw  &\qw &\qw & \ctrl{3} & \qw & \qw &\qw \\
	 \lstick{q_{2}} &  \qw & \qw & \qw & \qw & \qw &\qw
	 &  \qw   & \qw & \qw & \qw & \qw  &  \qw & \qw &\qw
	 & \qw & \qw &  \gate{H} & \ctrl{2} & \qw  & \ctrl{2} & \gate{H} &\qw  &\qw 
	 & \qw & \qw & \qw &\qw &\qw &\qw &\ctrl{2} & \qw & \qw & \qw &  \qw  & \qw  & \ctrl{2} & \qw & \qw &\qw &\qw&\qw \\
	 \lstick{q_{3}} &  \qw & \qw & \qw & \qw &  \qw & \qw  & \qw & \qw & \qw & \qw  &  \qw & \qw & \qw  & \qw & \qw & \qw & \qw  &  \qw & \qw & \qw & \qw  &  \qw & \qw
	 & \qw &\qw
	 & \qw & \qw &  \qw &  \qw &  \qw &\gate{H} & \ctrl{1} & \qw  & \ctrl{1} & \gate{H} &\qw  &\qw  & \qw & \qw & \qw &\qw\\
    \lstick{\ket{0}} & \qw & \qw\bigoplus &\gate{R_z(t)}& \bigoplus & \qw  & \qw \bigoplus & \qw & \qw \bigoplus & \gate{R_z(t)}   & \qw \bigoplus & \qw  & \qw \bigoplus &\qw & \qw \bigoplus & \qw  & \qw \bigoplus & \qw  &\qw \bigoplus & \gate{R_z(t)}   & \bigoplus & \qw  & \qw \bigoplus & \qw & \qw \bigoplus 
    &\qw & \qw \bigoplus & \qw  & \qw \bigoplus &\qw & \qw \bigoplus & \qw  & \qw \bigoplus   & \gate{R_z(t)}   & \qw \bigoplus & \qw  & \qw \bigoplus & \qw & \qw \bigoplus&\qw & \qw \bigoplus&\qw \\
    & & &\sigma_x\otimes I \otimes I \otimes  I &  &&&  && \sigma_z \otimes \sigma_x \otimes I \otimes I &  &  &&&   &&&  && \sigma_z \otimes \sigma_z \otimes \sigma_x \otimes I & &  &&&   &&  &&&   &&  && \sigma_z \otimes \sigma_z \otimes \sigma_z \otimes \sigma_x &  \\
	 }
	 \end{equation*}
	 \caption{Basic circuit to simulate $B$ with $n=4$.
	 }\label{fig:1}
\end{figure}
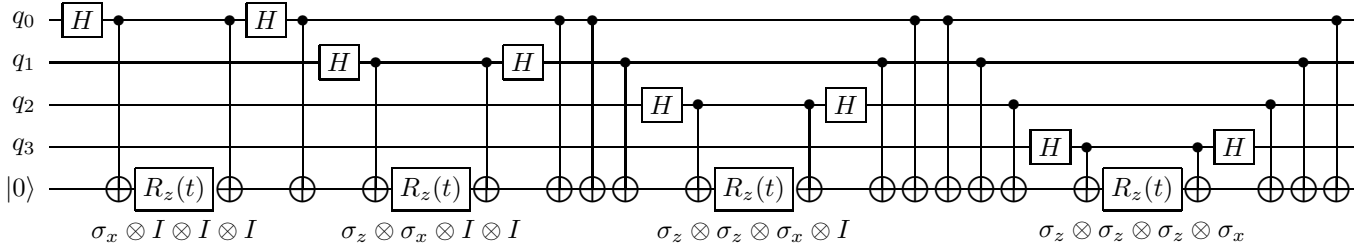

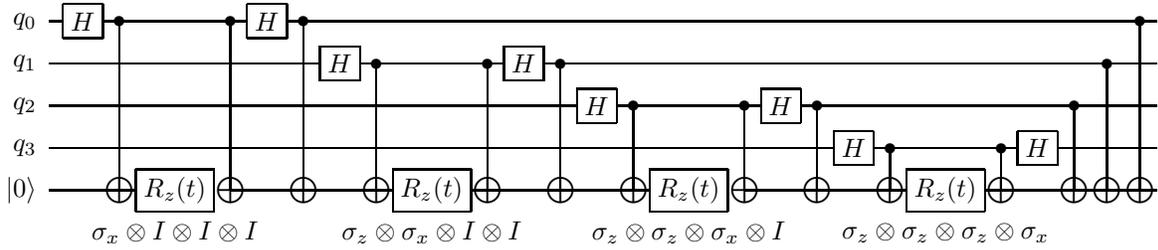
\begin{figure}[h]
\begin{equation*}
\Qcircuit @C=0.5em @R=0.0em @!R {
	 \lstick{q_{0}} &  \gate{H} & \ctrl{4} & \qw  & \ctrl{4} & \gate{H}    
	 & \ctrl{4} & \qw & \qw & \qw &  \qw  &  \qw                 
	  & \qw & \qw &  \qw  &  \qw &  \qw  &  \qw & \qw &\qw &\qw & \qw & \qw &  \qw  &  \qw &  \qw   &  \qw &  \qw  & \ctrl{4}&\qw \\
	 \lstick{q_{1}} &  \qw & \qw & \qw & \qw &  \qw 
	 & \qw &  \gate{H} & \ctrl{3} & \qw  & \ctrl{3} & \gate{H}   & \ctrl{3} & \qw & \qw & \qw &  \qw   & \qw & \qw & \qw & \qw  & \qw &\qw  &  \qw  &\qw &\qw & \ctrl{3} & \qw & \qw &\qw \\
	 \lstick{q_{2}} &  \qw & \qw & \qw & \qw & \qw &\qw
	 &  \qw   & \qw & \qw & \qw 
	 & \qw & \qw &  \gate{H} & \ctrl{2} & \qw  & \ctrl{2} & \gate{H} &\ctrl{2} & \qw  &\qw 
	 & \qw & \qw & \qw  &\ctrl{2} & \qw & \qw & \qw & \qw & \qw \\
	 \lstick{q_{3}} &  \qw & \qw & \qw   &  \qw & \qw & \qw & \qw  &  \qw & \qw & \qw & \qw  &  \qw & \qw
	 & \qw & \qw &  \qw &  \qw &  \qw &\gate{H} & \ctrl{1} & \qw  & \ctrl{1} & \gate{H} &\qw  &\qw  & \qw & \qw & \qw &\qw\\
    \lstick{\ket{0}} & \qw & \qw\bigoplus &\gate{R_z(t)}& \bigoplus & \qw  & \qw \bigoplus & \qw & \qw \bigoplus & \gate{R_z(t)}   & \qw \bigoplus & \qw   & \qw \bigoplus&\qw  &\qw \bigoplus & \gate{R_z(t)}   & \bigoplus & \qw    & \qw \bigoplus &\qw  & \qw \bigoplus   & \gate{R_z(t)}   & \qw \bigoplus & \qw  & \qw \bigoplus & \qw & \qw \bigoplus&\qw & \qw \bigoplus &\qw \\
    & & &\sigma_x\otimes I \otimes I \otimes I &  &&&&& \sigma_z \otimes \sigma_x \otimes I \otimes I &&&&&& \sigma_z \otimes \sigma_z \otimes \sigma_x \otimes I &&&&&& \sigma_z \otimes \sigma_z \otimes \sigma_z \otimes \sigma_x &  \\
	 }
\end{equation*}
	 \caption{Simplified circuit after gate cancellations. }\label{fig:2}
\end{figure}
We therefore obtain a simplified circuit for $B$ that has just $O(n)$ gates, as shown in Figure~\ref{fig:2}, using gate cancellations related to commuting Pauli subterms.
The simplified circuit for simulating $e^{iBt}$ now has $O(n)$ gates, i.e., $2(2n-1)$ CNOT gates, $2n$ Hadamard gates, and $n$ rotation gates. The  depth of this circuit is only $4n$, and therefore it is   NISQ implementable.

\subsection{Chebyshev Approximation}\label{app:chebyshev}
We now present details about Chebyshev polynomial approximations.
Given a function $f:[-1,1] \rightarrow \RR$, an $m$-degree Chebyshev polynomial approximation is then expressed as
     \[
     f(x) \approx p_m(x) = \sum_{j=0}^m c_j T_j(x),
     \]
    where  $T_j(x)$ is the $j$th-degree Chebyshev polynomial of the first kind, formally defined as 
$T_j (x) = \cos(j \cos^{-1} (x))$.
We therefore have $T_0(x) =1$, $T_1(x) = x$ and
    \[
    T_{j+1}(x) = 2x T_j(x) - T_{j-1} (x),
    \]
   and the (interpolation) coefficients are computed as
    \[
    c_j = \frac{2-\delta_{j0}}{\pi}\int_{-1}^1 \frac{f(x)T_j(x)}{\sqrt{1-x^2}}dx \qquad \text{ or } \qquad
     c_j = \frac{2-\delta_{j0}}{m+1}\sum_{k=0}^m f(x_k)T_j(x_k)
         \]
         with Chebyshev nodes $x_k = \cos\left(\frac{\pi(k+1)/2}{m+1}\right)$,
when interpolation is used. 
Supposing the given matrix $A$ has its spectrum in the interval $[\lambda_{\min},\lambda_{\max}]$, then the matrix  $\tilde{A} = \left(\frac{2A - (\lambda_{\max}+\lambda_{\min})I}{\lambda_{\max}-\lambda_{\min}}\right)$ has the spectrum in the interval $[-1,1]$.

Given a function $f(\cdot)$, we can compute the trace of the matrix function,
$\trace(f(A))$, using the stochastic Chebyshev method. 
 We approximate $\bra{v_l}f(A) \ket{v_l} \approx \bra{v_l} TB \ket{v_l}$, where $B = \sum_{j=0}^m \tilde{c}_j T_j(\tilde{A})$.
 Let $\ket{w_l^{(j)}} = T_j(\tilde{A})\ket{v_l}$ with
 $$ \ket{w_l^{(j+1)}} = 2\tilde{A}\ket{w_l^{(j)}} - \ket{w_l^{(j-1)}},$$
where $~\ket{w_l^{(0)}} = \ket{v_l}~$ and  ~$\ket{w_l^{(1)}}= \tilde{A}\ket{v_l}$. 
%
The matrix function trace can then be estimated as
 \begin{equation}\label{eq:rank}
\trace(f(A))_\approx \frac{1}{\nv}\sum_{l=1}^{\nv}\left[\sum_{j=0}^m\tilde{c}_j \bra{v_l}\ket{w_l^{(j)}}\right].
\end{equation}
The computational time complexity of the method will be $O(\nz(A)m\nv)$, where $\nz(A)$ is the number of nonzeros in the matrix $A$.
This method has been quite popular for estimating spectral density~\cite{wang1994calculating,lin2016approximating} (where it is known as the \emph{Kernel Polynomial Method}), for estimating the eigencount of a matrix~\cite{di2016efficient}, for numerical rank estimation~\cite{ubaru2016fast,ubaru2017fast}, and for other analytic function estimation and applications~\cite{han2017approximating}.
 
In this paper, we are interested in expanding the step function as
$
 h(t) \approx \sum_{j=0}^mc_jT_j(t).
$
The expansion coefficients $c_j$ for the polynomial to 
approximate a step function $h(t)$, 
taking value $1$ in $[a,\ b]$   and $0$ elsewhere, are known to be given by
\[
 c_j=\begin{cases}
\frac{1}{\pi}(\cos^{-1}(a)-\cos^{-1}(b)) & : \: j=0\\
                                  \frac{2}{\pi}\left(\frac{
                                  \sin(j\cos^{-1}(a))-\sin(j\cos^{-1}(b))}{j}\right)  & : \: j>0
                        \end{cases} \;\;.
\]
Since the step function is a discontinuous function, we use in our analysis a surrogate function $f(x) = \frac{1}{2}\left(1+\tanh(\alpha(x - \frac{\delta}{2}))\right)$, with $\alpha  = \frac{1}{\delta}\log(\frac{2}{\epsilon})$. 

\subsection{Moments Computation}\label{app:moments}
The proposed NISQ-QTDA Algorithm~\ref{alg:algo1} uses the quantum computer  only to calculate the moments of the Laplacian $\Delta_k$, i.e., $\mu_l^{(i)}$ for $i=0,\ldots,m$ and $l=1,\ldots,\nv$. Letting $\ket{v_l}$ be a random state vector, the $i$th moment is given by
\begin{align}
  \mu_l^{(i)} = \bra{v_l} \Delta_k^i \ket{v_l} &= \bra{v_l} (P_k \PG B \PG B \PG P_k)^i \ket{v_l} .
\end{align}
For a given $\{i,l\}$, suppose $\ket{\phi_l^{(i)}} = \left(\prod_{j=0}^{j=i-1} \PG P_k^{j  \textrm{\%} 2} B\right) \PG P_k \ket{v_l}$, where the binary modulo operation $\%2$ indicates that the projection $P_k$ is applied only in the odd rounds, since we will not have any components of order $k$ after an odd number of $B$ applications. We can compute the moments $\mu_l^{(i)}$ using two approaches. 
The first approach is to first compute the state $\ket{\phi_l^{(i)}}$, then compute its inverse (i.e., complex conjugate), and estimate/measure  $ \mu_l^{(i)}= \left\langle\phi_l^{(i)}\bigg|\phi_l^{(i)}\right\rangle$.  A drawback of this approach is that it requires us to prepare two quantum states (i.e., $\ket{\phi_l^{(i)}}$ and its inverse).
An alternative approach is to compute the state $\ket{\phi_l^{(i)}}$ and estimate its norm, given that $\mu_l^{(i)} = \left\|\ket{\phi_l^{(i)}}\right\|^2$. Since the state vector is of exponential size $2^n$, we will need to use a repeated counting technique to estimate its norm.


Both of these approaches reduce to preparing the quantum states  $\ket{\phi_l^{(i)}} = \left(\prod_{j=0}^{j=i-1} \PG P_k^{j  \textrm{\%} 2} B\right) \PG P_k \ket{v_l}$ for $i=0,\ldots,m$ and $l=1,\ldots,\nv$.
In Section~\ref{ssec:projection} of the main body of the paper, we presented methods to build the projectors $\PG$ and $P_k$ on a NISQ device. We also discussed the simulation of $B$ on a NISQ device in the previous section. Therefore, using a quantum circuit of $O(n)$ size (both gate and depth complexity),  we can compute $\ket{\phi_l^{(i)}}$ and the moments $\mu_l^{(i)}$ in order to estimate the Betti numbers $\chi_k$ using  our Algorithm~\ref{alg:algo1}.

\paragraph{Errors:} Note that the above quantum moment computation will incur a certain error in the estimation. This is because we only simulate the unitary operator $U_B(t) = e^{-iBt}$ in the quantum computer, and we have 
\[
U_B(t) = e^{-iBt} = 1 - iBt - B^2t^2/2 + O(Bt^3) = 1 - iBt - nIt^2/2 + O(Bt^3) ,
\]
as $t\rightarrow 0$.
Hence, say for $i=1$, what gets measured is given by
\begin{align*}
  M(\ket{\phi_l}) &= \bra{v_l} P_k \PG U_B \PG U_B \PG P_k\ket{v_l} \\
             & = \bra{v_l}P_k \PG (1 - iBt) \PG(1 - iBt) \PG P_k \ket{v_l}
                -n \bra{v_l} P_k \PG P_k \ket{v_l} O(t^2) + O(t^3) \\
             &\approx \bra{v_l} P_k \PG P_k \ket{v_l}(1-O(nt^2)) - 2it\bra{v_l} P_k \PG B \PG P_k \ket{v_l} - t^2 \bra{v_l} P_k \PG B \PG B \PG P_k \ket{v_l} .
\end{align*}
We therefore obtain
\begin{align}
  \bra{v_l} \Delta_k \ket{v_l} &\approx \textrm{Re}(\bra{v_l}P_k \PG P_k \ket{v_l}(1-O(nt^2)) - M)/t^2 .
\end{align}
Errors in the calculation of higher moments can be estimated in a similar fashion. However, since we can choose $t$ to be very small, the error in the moment computation will be negligible. We refer to the next section for further details.


\section{Proof of Theorem~\ref{theo:main}}\label{app:proof}
Here we present the proof of our main Theorem~\ref{theo:main}. To this end, we first need the following error analysis for stochastic trace estimation using random Hadamard vectors.
\paragraph{Random Hadamard vectors:} 
For the stochastic trace estimator discussed in the main body of the paper, we need to form the random  vector state $|v_l \rangle$. It turns out that random  Hadamard columns have been shown to perform well in practice for trace estimation~\cite{fika2017stochastic}. Sampling a random Hadamard state on a quantum computer is feasible using a short-depth circuit. When the random state vector $| v_l\rangle = | h_{c(l)}\rangle$ is some random Hadamard column with $c(l)$ defining the random index, then the estimate
$\langle h_{c(l)} | A | h_{c(l)}\rangle$ can be viewed as a uniform random sample 
of the transformed matrix $M = HAH^T$ with  the Hadamard matrix $H$, i.e., 
$\langle h_{c(l)} | A | h_{c(l)}\rangle  = \langle e_{c(l)} | M | e_{c(l)}\rangle$ where $| e_{l}\rangle$ are basis vectors. We can now use the analysis of unit vector estimators in~\cite{avron2011randomized} to obtain error bounds. In particular, we apply Theorem $16$ of Avron and Toledo~\cite{avron2011randomized}.

\begin{lemma}\cite[Theorem 16]{avron2011randomized}\label{lemma:hada}
Assume we are given a Hermitian matrix $A\in\RR^{N\times N}$, error tolerance $\epsilon\in(0,1)$ and probability parameter $\eta\in (0,1)$. Then, for random state vectors $|v_l\rangle = | h_{c(l)}\rangle$ as random Hadamard columns, $l=1,\ldots,\nv$, and for $\nv \geq \frac{r^2_H(A)\log(2/\eta)}{\epsilon^2}$ where $r_H(A) = \max_i A_{ii}$, we have 
\begin{equation}\label{sec:eq:epsdel}
    \pr\left(\left| \frac{1}{\nv}  \sum_{l=1}^{\nv} \langle v_l | A | v_l\rangle
    - \trace(A)\right| \leq \epsilon \cdot N\right) \geq 1-\eta \; .
\end{equation}
\end{lemma}
The proof follows from the arguments establishing Theorem 16 in~\cite{avron2011randomized}, where we set $t = \epsilon \cdot N$ in the Hoeffding's inequality, the samples take values in the interval $[0,\max_i M_{ii}]$, and we know $M_{ii} = N\cdot A_{ii}$ since $H$ has orthogonal columns with $\|h_i\|^2 = N$. 

\paragraph{Proof of Theorem~\ref{theo:main}:} We now recall our main theorem and present its proof. 
\addtocounter{theorem}{-1}
\begin{theorem}
Assume we are given the pairwise distances of any $n$ data points and the encoding of the corresponding $\veps$-close pairs, together with an integer $0\leq k\leq n-1$ and the parameters $(\epsilon,\delta,\eta)\in (0,1)$.
Further assume the smallest nonzero eigenvalue of the scaled Laplacian $\tilde{\Delta}_k$ is greater than or equal to $\delta$, and choose $\nv$ and $m$ such that
\[
\nv = O\left(\frac{\log(2/\eta)}{\epsilon^2}\right)
\qquad\qquad \mbox{and} \qquad\qquad
m \geq 
 \frac{\log\left(\frac{32\log(2/\epsilon)}{\pi\delta\epsilon}\right)}{\log\left(1+\frac{\pi\delta}{4\log(2/\epsilon)}\right)}.
\]
Then, the Betti number estimation $\chi_k \in[0,1]$ by Algorithm~\ref{alg:algo1} satisfies
 \[
  \left| \chi_k - \frac{\beta_k}{\dim\tH_k}\right| \leq \epsilon,
  \]
with probability at least $1-\eta$.
\end{theorem}
\begin{proof}
We assume the Laplacian $\Delta_k$ has eigenvalues in the interval $[0,1]$. Since the largest eigenvalue is bounded by $O(k)$~\cite{gunn2019review}, we can scale $\Delta_k$ by $O(1/k)$ to ensure the spectrum of the scaled $\tilde{\Delta}_k$ is in $[0,1]$. Our approach consists of applying the stochastic Chebyshev method to the function $f(x) = \frac{1}{2}\left(1+\tanh(\alpha(x - \frac{\delta}{2}))\right)$, with $\alpha  = \frac{1}{\delta}\log(\frac{2}{\epsilon})$, as an approximation of the step function $S(x)$, which takes the value $1$ in the interval $[\delta,1]$ and the value zero elsewhere.
  If the smallest nonzero eigenvalue of $\tilde{\Delta}_k$ is greater than $\delta$, then the rank of $\Delta_k$ is $\trace(S(\tilde{\Delta}_k))$. 
Let $p_m$ be an $m$-degree Chebyshev approximation of $f(x)$ and let $I_{\delta} = \{0,[\delta,1]\}$, the interval where the eigenvalues of $\tilde{\Delta}_k$ lie.

We first show that
\[
 \max_{x\in I_{\delta}} |S(x) - p_m(x)|\leq \epsilon \;\;\; .
\]
We know
\[
 \max_{x\in I_{\delta}} |S(x) - p_m(x)|\leq  \max_{x\in I_{\delta}} |S(x) - f(x)|+ \max_{x\in I_{\delta}} |f(x) - p_m(x)|,
\]
and thus we need to show that the two terms are bounded by $\epsilon/2$.
For the first term, with  $f(x) = \frac{1}{2}\left(1+\tanh(\alpha(x - \frac{\delta}{2}))\right)$ and upon choosing $\alpha  = \frac{1}{\delta}\log(\frac{2}{\epsilon})$, we have
\begin{eqnarray*}
    \max_{x\in I_{\delta}} |S(x) - f(x)| &= &\max_{x\in I_{\delta}} \frac{1}{2}\left|1 - \tanh\left(\alpha\left(x - \frac{\delta}{2}\right)\right)\right|\\
   & = &\frac{1}{2}\left(1 - \tanh( \frac{\alpha\delta}{2})\right)\\
   & = &\frac{e^{-\alpha\delta}}{1+e^{-\alpha\delta}}\\
   & \leq &e^{-\alpha\delta} = \frac{\epsilon}{2}.
\end{eqnarray*}
For the second term, we use Theorem 1 of Han et al.~\cite{han2017approximating} which provides an error analysis for the Chebyshev approximation of functions that are analytic over a Bernstein ellipse. In addition, Theorem 11 of~\cite{han2017approximating} considers a similar $\tanh$ function. Adapting the proof of Han et al.\ to our setting, we obtain
\[
\max_{x\in [-1,1]} |f(x) - p_m(x)| \leq \frac{16\alpha}{\pi(1+\frac{\pi}{4\alpha})^m}.
\]
Hence, by setting $m \geq \frac{\log(32\alpha)-\log(\pi\epsilon)}{\log(1+\frac{\pi}{4\alpha})}$, we have
$\max_{x\in [-1,1]} |f(x) - p_m(x)| \leq \frac{\epsilon}{2}$, and thus we obtain the required bound of $\max_{x\in I_{\delta}} |S(x) - p_m(x)|\leq \epsilon$.

Next, note that the size of $\tilde{\Delta}_k$ is $\dim\tH_k$, and the Betti number is given by $\beta_k = \dim\tH_k - \trace(S(\tilde{\Delta}_k))$. The stochastic Chebyshev method approximates the trace as  $\trace_{\nv}(p_m(\tilde{\Delta}_k))= \frac{1}{\nv}\sum_l\langle v_l | p_m(\tilde{\Delta}_k) | v_l \rangle $, for random vector states $| v_l \rangle$. The Betti number estimate $\chi_k$ is then given by $\chi_k = 1 - \frac{\trace_{\nv}(p_m(\tilde{\Delta}_k))}{\dim\tH_k}$. We therefore need to bound
\[
\left| \chi_k - \frac{\beta_k}{\dim\tH_k}\right|  = \frac{1}{\dim\tH_k} \left|\trace_{\nv}(p_m(\tilde{\Delta}_k)) - \trace(S(\tilde{\Delta}_k))\right|. 
\]
Now, we have
\[
\left|\trace_{\nv}(p_m(\tilde{\Delta}_k)) - \trace(S(\tilde{\Delta}_k))\right| \leq
\left|\trace_{\nv}(p_m(\tilde{\Delta}_k)) - \trace(p_m(\tilde{\Delta}_k))\right|+
\left|\trace(p_m(\tilde{\Delta}_k)) - \trace(S(\tilde{\Delta}_k))\right|.
\]
From Lemma~\ref{lemma:hada}, since the maximum diagonal entry of $\tilde{\Delta}_k$  is $O(1)$ and $|p_m(x)|\leq 1$ for $x\in(0,1)$, we obtain for $\nv = O(\frac{\log(2/\eta)}{\epsilon^2})$
\[
\left|\trace_{\nv}(p_m(\tilde{\Delta}_k)) - \trace(p_m(\tilde{\Delta}_k))\right| \leq \epsilon \cdot  \dim\tH_k \;\; . 
\]
Moreover, we have from above
\begin{eqnarray*}
  \left|\trace(p_m(\tilde{\Delta}_k)) - \trace(S(\tilde{\Delta}_k))\right| & = & 
  \sum_{i=1}^{\dim\tH_k} |S(\lambda_i) - p_m(\lambda_i)|\\
  &\leq &  \dim\tH_k \cdot \max_{I_{\delta}} |S(x) - p_m(x)|\\
  &\leq & \epsilon \cdot  \dim\tH_k\;\; .
\end{eqnarray*}

\paragraph{ Additional error terms:} Finally, we note that there are two additional errors which occur in our algorithm, namely: (1) error in the simulation of $U_B = e^{iBt}$ (which is $O(n^2t^2)$, as discussed in Section~\ref{app:BOsim}); and (2) error in the estimation of the moments $\mu_l^{(i)} = \bra{v_l}\tilde{\Delta}_k^i \ket{v_l}$ (which is $O((i)\cdot t^2)$, as shown in Section~\ref{app:moments}). Both of these errors are additive in nature.  We therefore can combine the errors in our moment estimate using the quantum algorithm, say for $\hat{\mu}_l^{(i)}$, and write
\[
\left| \hat{\mu}_l^{(i)}- \mu_l^{(i)}\right| \leq (i)\cdot \hat{\eps},
\]
where $\hat{\eps} = O(n^2t^2)$ is the combined error.  The error in the Chebyshev moments estimate $\hat{\theta}_l^{(j)}$ is then given by $\left| \hat{\theta}_l^{(j)}- \theta_l^{(j)}\right| \leq C_{T_j} j \hat{\eps},$ where $C_{T_j}$ is some constant depending on the coefficients  of the Chebyshev polynomial $T_j$ of degree $j$ (with respect to monomials). Note that $C_{T_j}$ will be small since $|T_j(x)| \leq 1$ for $x\in[-1,1]$.
Hence, since the coefficients $c_j$ are of the form given in Section~\ref{app:chebyshev}, the error in the Betti number estimate $\hat{\chi}_k = 1-\frac{1}{\nv\dim \tH_k}\sum_{l=1}^{\nv}\left[\sum_{j=0}^m c_j\hat{\theta}_l^{(j)}\right]$ is given by
\begin{equation}\label{eq:add_err}
\left| \hat{\chi}_k- \chi_k\right| \leq \frac{Cm \hat{\eps}}{\dim \tH_k},
\end{equation}
where $C = \frac{2}{\pi}\max_j[{C_{T_j}(\sin(j\cos^{-1}(1))-\sin(j\cos^{-1}(\delta)))}]$ is some  constant.
Since  a choice of small $t\ll 1$ suffices for the simulation (and is also a requirement for the asymptotic error bound to hold; see Section~\ref{app:BOsim}), and since we are interested in simplices/clique-dense complexes for which $\dim\tH_k$ will be large, we then have $\hat{\eps}\ll \sqrt{\frac{\eps\cdot \dim\tH_k}{Cm}}$. Therefore, the above error in the estimation of $\hat{\chi}_k$ will be negligible, and the right-hand side of~\eqref{eq:add_err} will be $\ll\eps$.

In addition to these errors, during the actual hardware implementation, we will encounter additional errors due to noise. Two sources of noise exist, namely (1) shot noise due to measurement and (2) hardware noise. The shot noise is typically modelled using a Gaussian assumption, and hence is assumed to reduce as $O(1/\sqrt{T})$ for $T$ repeated measurements. The hardware noise is much more difficult to characterize, since it is hardware and technology dependent. Therefore, in order to account for errors due to these two sources of noise, we will need to repeat the whole experiment several times and draw statistics to compute the Betti numbers. 

Combining the results yields the desired bound in the theorem.
\end{proof}

\end{document}